\documentclass[11pt,reqno]{amsart}
\usepackage{amsmath,amssymb,amsthm}
\usepackage[dvipsnames]{xcolor}
\usepackage{graphicx}
\usepackage{fancyvrb}
\usepackage[colorlinks=true,linkcolor=blue,citecolor=blue]{hyperref}
\usepackage[margin=1.15in]{geometry}

\definecolor{editblue}{rgb}{0.05,0.15,0.75}

\definecolor{delgray}{rgb}{0.45,0.45,0.45}
\IfFileExists{ulem.sty}{%
  \usepackage[normalem]{ulem}%
}{%
}

\newtheorem{theorem}{Theorem}
\newtheorem{proposition}[theorem]{Proposition}
\newtheorem{lemma}[theorem]{Lemma}
\newtheorem{corollary}[theorem]{Corollary}
\theoremstyle{definition}
\newtheorem{definition}[theorem]{Definition}
\newtheorem{remark}[theorem]{Remark}
\newtheorem*{theoremA}{Theorem A}
\newtheorem*{theoremB}{Theorem B}
\newtheorem*{theoremC}{Theorem C}

\DeclareMathOperator{\Vol}{Vol}
\newcommand{\R}{\mathbb{R}}
\newcommand{\Z}{\mathbb{Z}}
\newcommand{\Nat}{\mathbb{N}}
\newcommand{\vx}{\vec{x}}
\newcommand{\vk}{\vec{k}}
\newcommand{\va}{\vec{a}}
\newcommand{\vq}{\vec{q}}
\newcommand{\vv}{\vec{v}}
\newcommand{\vW}{\vec{W}}
\newcommand{\Lap}[0]{\mathcal{L}}
\newcommand{\Four}[0]{\mathcal{F}}

\title{Extended Body Propagators}
\author{Rory O'Dwyer}
\address{Department of Physics, Stanford University}

\begin{document}

\begin{abstract}
Three seemingly distinct
extended-body propagators devised by Ansoldi et al., Erbin et al., and Stanford et al.\
respectively arise in completely different settings, though each can be considered a realization of the Polyakov
path integral. This paper proves that the Ansoldi propagator corresponds to the characteristic function of the random variable with the probability density
function $f(A)=C\,A\,(\alpha^{2}+A^{2})^{-D/2}\mathbf{1}_{A>0}$ in ambient dimension $D$, of scale
$\alpha=\sqrt{b/2}$. The exponent is fixed by the dimension and by nothing else in the Ansoldi/geometric comparison; the on-shell two-point amplitude of Erbin et al. and JT gravity agree with its $D=4$ member in the senses made precise below. After an extensive literature review,
the author believes that this observation and constructed area measure is novel. The area measure will, however, emerge as the natural candidate for the extension of the earlier non-extended body path space measures to two
dimensional simplicial complexes. We further show that the area measure of this paper is exactly
the law of the area of a triangle, two of whose sides are drawn from the measure
obtained in the companion manuscripts for the one-dimensional Polyakov integral. We close with an argument, rigorous in part, which reduces the measure of tubular worldsheets between two arbitrary
boundary curves to that of this triangle.
\end{abstract}

\maketitle

\section{Introduction}

In \cite{ODwyerQFT} and \cite{ODwyerLattice}, the point correlators of a very general class of
perturbative QFT were obtained from a measure over the set of relevant Feynman diagrams
embedded in $\R^{d}\times\R$. Those results have since been refined; the refinement which
bears on this paper is the algebraic proof of Theorem C given in
Section~\ref{sec:thmCproof}, which replaces the sampling argument used there. This measure was
derived from a novel mathematical object referred to as the `continuum multinomial coefficient' (\cite{WVLR} and \cite{CanoDiaz}); just as
the multinomial coefficient counts lattice paths between two points in $\Z^{d}$, the
continuum multinomial coefficient measures the set of directed paths between two points
in $\R^{d}$. Some version of the continuum multinomial measure has existed since the 1940s;
a similar scheme was used by Feynman in his lattice checkerboard derivation of the
$1+1$ dimensional Dirac propagator \cite{FeynmanHibbs}. As this formulation can
successfully replicate Feynman rules and interacting propagators in QFT, an immediate
question becomes posed: can an analogous measure derive String and M Theory?

The present paper isolates one purely computational question on which any such derivation
must turn, and answers it. The question is: what is the law of the random area? In
\cite{ODwyerQFT}, an initial effort was made to demonstrate that the marginal spherical
distribution $(\tau^{2}-I^{2})^{\frac{d-3}{2}}$ arose as a measure of paths of length $I$
between points displaced by an interval $\tau$, and that a transform of this measure obtains
the Klein--Gordon scalar propagator $\langle\phi\phi\rangle$. The word transform must be
read carefully here, and the point is not cosmetic, so we record it once in
Section~\ref{sec:reflection}: it is not the marginal spherical measure whose Laplace
transform is the propagator, but that measure differentiated $d-2$ times in the length
variable. The required derivative result is exactly the shape of the companion Theorem C
quoted below, in which a product of $p$ first order operators
$m s_{p}+\va_{j}\!*\!\nabla$ stands outside the continuum multinomial. These derivatives
belong to the one dimensional propagator statement and will remain outside the surface
pushforward below. The central comparison of this paper is instead
between the surface measure and the contour-space propagator of Ansoldi et al.; the conventional
bosonic-string two-point amplitude and the JT expression provide independent four-dimensional
checks, but are not imposed as arbitrary-dimensional constraints on the geometry.

\subsection{Extended Body Propagators}\label{sec:ebp}
In this section, we will describe three wholly disparate extended body QFTs that all seem
to emerge from the same random variable over extended body volumes. First, we state the
original heuristic of Polyakov. Let $p<d\in\Nat$, let $C_{1}$ and $C_{2}$ be disjoint
surfaces in $\R^{d}$, and let $\Omega=\{\partial M=C_{1}\sqcup C_{2}\mid M\text{ is a
}p\text{-dimensional manifold in }\R^{d}\}$. Then the heuristic Polyakov path integral
between $C_{1}$ and $C_{2}$ is given in Equation~\eqref{eq:polyakov}.
\begin{equation}\label{eq:polyakov}
\sum_{M\in\Omega}e^{iT\mu(M)}
\end{equation}
In Equation~\eqref{eq:polyakov}, $\mu(M)=\mathcal{H}^{p}(M)$ is the $p$-dimensional
Hausdorff measure of $M$, that is, the volume induced on $M$ by the ambient Euclidean
metric. $T$ is a real-valued parameter known as the ``brane tension''. From
Equation~\eqref{eq:polyakov}, one may immediately understand the Polyakov path integral to
be the characteristic function of the measure of a manifold-valued random variable. Let us
now restrict ourselves to the setting where $C_{1}$ and $C_{2}$ are one-manifolds. We will
motivate what probability density function (pdf) should be carried by the random area in
order to replicate the Polyakov path integral. We do this through introducing the integral
in multiple physics applications. Consider the three seemingly distinct ``string''
propagators devised by Ansoldi et al.\ \cite{Ansoldi}, Erbin et al.\ \cite{Erbin}, and
Stanford et al.\ \cite{SYY} respectively in Equation~\eqref{eq:three}.
\begin{equation}\label{eq:three}
\int_{0}^{\infty}dW\,e^{-\frac{im^{2}W}{2}}\Big(\frac{m^{2}}{2i\pi W}\Big)^{\frac32}
e^{\frac{im^{2}b}{4W}},\qquad
\sqrt{\vk^{2}-|s|T},\qquad
\int ds\,\frac{s}{2\pi^{2}}\sinh(2\pi s)\,e^{-\beta\frac{s^{2}}{2}}
\end{equation}
It is beyond the scope of this work to delve into all parameters of
Equation~\eqref{eq:three}. The three expressions arise in different reductions of extended-body
physics. At $D=4$ they meet on the Fisk law in the senses made precise in
Sections~\ref{sec:ansoldi}, \ref{sec:tachyon} and~\ref{sec:jt}. Only the Ansoldi loop-space
kernel is used below as the arbitrary-dimensional contour-to-contour comparison: its area density is
the family $A(\alpha^{2}+A^{2})^{-D/2}$. The conventional tachyon two-point amplitude has a
different general-$D$ area tail and is therefore treated as a separate observable rather than as a
dimensional constraint on the geometric measure. In the four-dimensional Fisk case, what is
matched is the odd part of the characteristic function, equivalently the Fourier transform of the odd
extension of Equation~\eqref{eq:fisk}; since $A>0$ this determines the law of $A$ uniquely, and in
imaginary time the same pairing is the Laplace transform of Section~\ref{sec:reduction}. The moment
generating function does not exist for positive argument, the Fisk tails being polynomial.
\begin{equation}\label{eq:fisk}
f(A)=C\,\frac{A}{(\alpha^{2}+A^{2})^{2}}\,\mathbf{1}_{A>0}
\end{equation}
After an extensive literature review, the author believes that this observation is novel.
It will, however, emerge as the natural $D=4$ member of the extension of the author's path-space
measure to two dimensional simplicial complexes. The exponent $2$ in Equation~\eqref{eq:fisk} is
not universal. The Ansoldi loop-space
kernel carries $A(\alpha^{2}+A^{2})^{-D/2}$ by Corollary~\ref{cor:generalD}, and the fixed-side
geometry of Section~\ref{sec:triangle} produces the same exponent dimension for dimension through
Lemma~\ref{lem:mshift}. The conventional tachyon expression agrees with its degenerate tail at
$D=4$ but has a different general-$D$ tail, as Section~\ref{sec:tachyon} makes explicit. Equation~\eqref{eq:fisk}
is the Fisk, or log logistic, density of shape parameter $2$ and scale parameter $\alpha$; we will
call it the Fisk density throughout, and we normalise it in Section~\ref{sec:fisklaw}. In the
Ansoldi comparison the scale is fixed by the boundary data as $\alpha=\sqrt{b/2}$, where
$b=(\Sigma^{\mu\nu}(\sigma_{j})-\Sigma^{\mu\nu}(\sigma_{i}))(\Sigma_{\mu\nu}(\sigma_{j})
-\Sigma_{\mu\nu}(\sigma_{i}))$ is the invariant built from the areal coordinates
$\Sigma^{\mu\nu}(\sigma)=\int_{\sigma}x^{\mu}dx^{\nu}$ of the two boundary loops. We
reserve $\alpha$ for the scale of the Fisk law and $b$ for this areal invariant.

The quantity
$b^{\mu\nu}$ is a \emph{signed} shadow area: by Stokes' theorem it is the area of the
projection of an interpolating surface onto the $\mu\nu$ plane counted \emph{with
orientation}, so that a fold contributes with opposite sign to the sheet it folds back over.
It is therefore not the area of the extremal surface. In $\R^{3}$ the three numbers
$b^{23},b^{31},b^{12}$ are the components of the familiar vector area
$\vec{S}=\int_{K}\hat{n}\,dA$ of elementary vector calculus, and $\sqrt{b/2}=|\vec{S}|$.
What is true, and is proved as Proposition~\ref{prop:signedbound} below, is the inequality
\[
\sqrt{b/2}\ \le\ \mu_{\min}(\sigma_{i},\sigma_{j}),
\qquad\text{with equality exactly when the extremal surface is flat.}
\]
The scale carried by the construction of this paper is the extremal area
$\mu_{\min}$, and it is on the compact branch of Proposition~\ref{prop:branches} that this is
transparent: there the density $A(\alpha^{2}-A^{2})^{-2}$ is supported on $A<\alpha$ and diverges at
the endpoint, so the scale is the extremal area available to the configuration and is carried by the
boundary data rather than fitted, the Fisk scale being the image of that endpoint under
$\alpha^{2}\mapsto-\alpha^{2}$. The scale carried by the loop space propagator of \cite{Ansoldi} is $\sqrt{b/2}$; and the two agree on the flat regime, which is
the regime of Section~\ref{sec:tachyon}, of Section~\ref{sec:degenerate}, and of the areal
reduction in which \cite{Ansoldi} works. Remark~\ref{rem:quench} states this relationship,
and nothing in the Polyakov heuristic of Equation~\eqref{eq:polyakov} is affected by it,
that sum being taken at fixed contours and weighted by the unsigned area.

The plan of the paper is as follows. Section~\ref{sec:cm} quotes, without proof, the results
of the companion manuscripts which motivated this work: the continuum multinomial
coefficient, the clean Polyakov expression it yields for the scalar propagator, and the
exact sense in which that expression involves derivatives of the marginal spherical measure.
It closes by setting out the dictionary between that one dimensional construction and the
two dimensional one pursued here. Section~\ref{sec:fisklaw} fixes notation for the Fisk law,
restates the $q$-Tsallis objects needed to name it in the form used in \cite{ODwyerQFT}, and
proves the two transform identities on which everything else rests.
Sections~\ref{sec:ansoldi}, \ref{sec:tachyon} and~\ref{sec:jt} treat the three expressions of
Equation~\eqref{eq:three} in turn. Section~\ref{sec:triangle} shows that the same density is
exactly the law of the area of a triangle whose sides are drawn from the one dimensional
Polyakov measure of \cite{ODwyerQFT}. Section~\ref{sec:reduction} collects, and makes
rigorous as far as the author is able, the argument which reduces an arbitrary pair of
boundary curves to that triangle. Section~\ref{sec:thmCproof} proves Theorem C.

\section{The Continuum Multinomial and the One Dimensional Polyakov Measure}\label{sec:cm}

This section quotes, without proof, the results of the companion manuscripts which motivated
the present paper. Nothing in it is new. The proof of
Theorem C, which is the only quoted result used below, is reproduced in
Section~\ref{sec:thmCproof}; the proofs of Theorems A and B are not repeated. It is included for two
reasons. The first is that the object of Sections~\ref{sec:ansoldi}--\ref{sec:triangle} is
the two dimensional analogue of the measure quoted here, and the analogy is only legible
once the one dimensional statement is on the page. The second is that the correction
recorded in Section~\ref{sec:reflection} concerns the precise sense in which that measure
produces a propagator, and that correction is stated in the notation fixed below.

As defined by Cano and Diaz in \cite{CanoDiaz}, and later generalized by Wakhare, Vignat, Le,
and Robins in \cite{WVLR}, we consider the continuous multinomial coefficient in
Equation~\eqref{eq:cm} for $l$ variables. Let $\{x_{i}\}_{i=1}^{l}\subset\R_{+}$, let $\mu$
denote the Lebesgue measure on $\R^{n}$ \cite{Folland}, and for $l,N\in\Nat$ denote by
$D(N,l)$ the set of words of length $N$ drawn from $l$ distinct letters in which adjacent
letters are distinct \cite{WVLR}. This set is referred to in the combinatorics literature as
the set of Smirnov words. Furthermore let $c\in D(n,l)$, let $\{c_{k}\}_{k=1}^{n}$ denote the
letters of our Smirnov word and let $n$ denote its length, taking elements among some
collection of $l$ dimensional vectors, let $\{e_{k}\}_{k=1}^{l}$ denote these vectors, and let
$\{\eta_{k}\}_{k=1}^{l}$ denote the multiplicity of direction $c_{k}$. Note that $k$ runs from
$1$ to $n$ to denote order; there are only $l$ distinct letters while there are $n$ letters
total. Where the two indices must be distinguished we reserve $l$ for the number of
directions and $d$ for the dimension of the ambient space. Consider the partial ordering
imposed on $\R^{d}$ such that $\vx\le\vec{y}$ iff $\vec{y}-\vx\in\R^{d}_{+}$. Then, we define
the path polytope $P(\vq,c)$ in Equation~\eqref{eq:polytope} for $\vq\in\R^{d}_{+}$.
\begin{equation}\label{eq:polytope}
P(\vq,c)=\Big\{(\lambda_{1},\dots,\lambda_{n})\in\R^{n}_{+}\ \Big|\
\sum_{k=1}^{n}\lambda_{k}e_{c_{k}}\le\vq\Big\}
\end{equation}
This obtains the definition of the continuum multinomial for $\{x_{i}\}_{i=1}^{l}\subset\R$ in
Equation~\eqref{eq:cm}:
\begin{equation}\label{eq:cm}
\binom{\sum x_{i}}{x_{1},x_{2},\dots,x_{l}}=\sum_{n=0}^{\infty}\sum_{c\in D(n,l)}
\mu\big(P(\vq,c)\big),\qquad \vq=\sum x_{i}e_{i}
\end{equation}
One may alter the path polytope of Equation~\eqref{eq:polytope} such that each path
necessarily terminates in $\vq$ (by substituting $\le$ with $=$). This would lead, in
Equation~\eqref{eq:cm}, to another convergent and integrable function described in
\cite{WVLR}; we will denote this object with the superscript \emph{thin}.

The direction set $\{\hat{e}_{i}\}_{i=1}^{l}$ implicit in Equation~\eqref{eq:polytope} is one
instance of a structure that recurs throughout the companion work. The companion result
quoted below replaces it by a large set of directions equidistributed on a sphere. We
therefore fix the following language once and for all, and phrase the rest of this section in
it.

\begin{definition}[Direction system]\label{def:dirsys}
A \textbf{direction system} is a finite set $A\subset S^{d-1}$. Its lift is
$\hat{A}=\{(\va,1)\mid\va\in A\}\subset\R^{d}\times\R$, where the last coordinate is written
$\ell$ and called the length coordinate. The step cone of $A$ is
\[
C_{A}=\Big\{\sum_{\va\in A}v_{\va}(\va,1)\ \Big|\ v_{\va}\ge0\Big\}\subset\R^{d}\times\R,
\]
and we write $u\le_{A}w$ if and only if $w-u\in C_{A}$. A path is $A$-monotonic at spacing
$\epsilon$ if it is piecewise linear with every step in $\epsilon\hat{A}$. We call $A$
isotropic if it satisfies the two conditions of Equation~\eqref{eq:isotropy} below.
\end{definition}

For the special case $A=\{\hat{e}_{i}\}_{i=1}^{d}$ one has $\ell=|\vx|_{1}$ identically on
$C_{A}$, so the length coordinate carries no information beyond $\vx$ and $\le_{A}$ reduces to
the partial order on $\R^{d}$ described above. It is only when $A$ is larger that $\ell$
becomes an independent coordinate.

The following results are taken from a companion manuscript currently under consideration.
Their proofs are not repeated here. They are quoted as Theorems A, B and C so as not to
collide with the numbering of the results proved in this paper. Theorems A and B are quoted
only to motivate the construction and are not used below; Theorem C is used in
Section~\ref{sec:reflection}.

\begin{theoremA}
Let $t\in\R_{+}$ bound the upper argument of the continuum multinomial coefficient (i.e.\
$\sum x_{i}\le t$). Then the continuum multinomial coefficient is supported within the domain
$\{(x_{1},\dots,x_{l})\mid\sum x_{i}\le t,\ x_{i}\ge0\}$ and its definition in
Equation~\eqref{eq:cm} has a Taylor series with infinite radius of convergence and thereby
converges uniformly and absolutely on compact sets in its domain. This implies that the
continuum multinomial is integrable over this domain, and hence, once extended by zero
outside it, over $\R^{l}$.
\end{theoremA}

Note that Theorem A proves that the continuum multinomial has an infinite radius of
convergence. Therefore, the function uniformly converges on compact subsets. Since the
fixed-$t$ domain is compact and the function is continuous, the continuum multinomial is
bounded and Lebesgue integrable. This is required if we want the continuum multinomial to be
the pdf of some random vector $\vec{Z}$.

For the purpose of this article, every time we mention monotonic paths, it will be with
respect to some lattice spacing $\epsilon$. We define monotonic paths as piecewise linear
paths in $\epsilon\Z^{l}$ whose steps are in $\{\epsilon\hat{e}_{i}\}_{i=1}^{l}$, i.e.\
positively oriented vectors with length $\epsilon$ in the scaled lattice $\epsilon\Z^{l}$. The
function $g_{\epsilon}$ in Theorem B below is proportional to the pdf of a random variable
related to these paths. Note that the domain of the pdf in this theorem, $D_{t}$, is a simplex
with a bounded major axis $t$, or $D_{t}=\{\vx\in\R^{l}_{+}:\sum_{i=1}^{l}x_{i}=t\}$.

\begin{theoremB}
Let $\epsilon,t>0$. Let $g_{\epsilon}(x):D_{t}\to\R$ be defined as in
Equation~\eqref{eq:geps}.
\begin{equation}\label{eq:geps}
g_{\epsilon}(x)=\sum_{n=0}^{\infty}\sum_{c\in D(n,l)}\epsilon^{n}
\left|\left\{(\lambda_{1},\dots,\lambda_{n})\in\epsilon\Z^{n}_{+}\ \Big|\
\vx-\sum_{k=1}^{n}\lambda_{k}\hat{e}_{c_{k}}\in\R^{l}_{+}\right\}\right|.
\end{equation}
Then we have pointwise
\[
\lim_{\epsilon\to0}g_{\epsilon}(x)=\binom{\sum_{i=1}^{l}x_{i}}{x_{1},\dots,x_{l}}
\]
on $D_{t}$, and they are both positive and have finite integral.
\end{theoremB}

Theorem B establishes the continuum multinomial as a measure over paths that terminate prior
to some endpoint. More specifically, it demonstrates how one would find the measure of a more
general collection of paths, namely through use of Theorem B's path reweighting scheme. If
one is familiar with the Feynman checkerboard model \cite{FeynmanHibbs} of the $1+1$d
Dirac propagator, one may immediately recognize that it yields the same weighting of paths as
in Equation~\eqref{eq:geps}. I.e.\ the continuum multinomial coefficient is a generalization
of the Feynman checkerboard measure to arbitrary-dimensional scalar settings.

With the continuum multinomial established as a well-defined function/measure over continuum
paths, we now establish its relationship to the Polyakov propagator. As mentioned in the
introduction, we demonstrate its relationship to derivatives of the uniform distribution in
the sphere. For each $p\in\Nat$, consider a choice of $p$ rays in $\R^{d}$ whose directions
are equidistributed along $S^{d-1}$. Let $A_{p}$ denote the unit vectors along these rays,
that is, the points at which these rays intersect $S^{d-1}$. Furthermore, let
$\{\va_{i}\}_{i=1}^{p}$ denote an indexing of these vectors. We require in addition that
$A_{p}$ satisfy the two exact conditions
\begin{equation}\label{eq:isotropy}
\sum_{j=1}^{p}(\vk*\va_{j})^{n}=0\ \text{ for every odd }n\text{ and every }\vk\in\R^{d},
\qquad
\sum_{j=1}^{p}(\vk*\va_{j})^{2}=\frac{p}{d}|\vk|^{2}\ \text{ for every }\vk\in\R^{d}.
\end{equation}
The first holds whenever $A_{p}$ is closed under $\va\mapsto-\va$. Both hold, for instance, if
one takes $\frac{p}{2d}$ copies of an orthonormal basis of $\R^{d}$, each rotated arbitrarily,
together with the negatives of all of those vectors. These two conditions are what the proof
uses; equidistribution alone is not enough for the first, since the $n=1$ term needs
$\sum_{j}\va_{j}$ to vanish rather than merely to be small compared to $p$.

\begin{theoremC}
Let $G(m,\vx)$ denote the solution to the Klein--Gordon propagator in Euclidean dimension
$d$, that is, the kernel of $(-\Delta+m^{2})^{-1}$ on $\R^{d}$, given explicitly by
$G(m,\vx)=(2\pi)^{-d/2}(m/|\vx|_{2})^{\frac{d}{2}-1}K_{\frac{d}{2}-1}(m|\vx|_{2})$ for
$\vx\ne0$. Let $\{A_{p}\}_{p}$ be isotropic direction systems in the sense of
Definition~\ref{def:dirsys}, with $\{\va_{j}\}_{j=1}^{p}$ an indexing of $A_{p}$. Then,
pointwise for $\vx\ne0$,
\begin{equation}\label{eq:thmC}
\begin{split}
G(m,\vx)=\lim_{p\to\infty}\frac{1}{dm^{2}p^{2}}\prod_{j=1}^{p}
\Big(m\big(p-1+\tfrac{1}{dp}\big)+\va_{j}*\nabla\Big)\times\\
\int_{\R^{p}_{+}}e^{-m(p-1+\frac{1}{dp})\sum_{j}v_{j}}
\binom{m\sum_{j}v_{j}}{mv_{1},\dots,mv_{p}}
\delta\Big(\vx-\sum_{j=1}^{p}v_{j}\va_{j}\Big)\,d\vv.
\end{split}
\end{equation}
\end{theoremC}

A proof of Theorem C is given in Section~\ref{sec:thmCproof}. It is algebraic throughout,
and it uses Theorem A but not Theorem B. The last equation in Theorem C demonstrates a rigorous, lattice-free realisation of the
worldline representation of the scalar propagator to which Polyakov's heuristic refers.
Namely, we have defined a genuine measure of rectifiable paths, and the relativistic
propagator is a transform from the length of those paths to mass as its conjugate variable.
It is necessary to note that we do require a $p\to\infty$ limit. The author believes this
limit is forced, as the limit is used several times throughout the proof of Theorem C
(Section~\ref{sec:thmCproof}) in a manner seemingly
inaccessible to a finite list of directions.
The author has demonstrated that the Polyakov measure is carried by Feynman-like measures
with arbitrary precision.

\subsection{Notation}\label{sec:notation}
The following symbols are fixed for the remainder of the paper. Throughout, $d$ is the
dimension of the ambient Euclidean space, $m$ is a mass, $\vx$ is a displacement and
$r=|\vx|_{2}$, $L$ is the length of a path, $A$ is the area of a surface, $T$ is the brane
tension, $\alpha$ is the scale parameter of a Fisk law, $q$ is the Tsallis index, $\beta$ is
an inverse temperature, $\nu=\frac{d}{2}-1$ is the Macdonald index, and $b$ is the areal
invariant of a pair of boundary loops defined in Section~\ref{sec:ebp}. We write
$\Four|^{T}_{A}$ for the Fourier transform carrying the variable $A$ to the variable $T$ and
$\Lap|^{m}_{L}$ for the Laplace transform carrying $L$ to $m$.

Three notational collisions are worth flagging, as three of them are unavoidable if the
companion results are to be quoted verbatim. The letter $n$ denotes the length of a Smirnov
word in Equation~\eqref{eq:cm} and the order of a derivative in Section~\ref{sec:reflection};
the letter $\mu$ denotes the Lebesgue measure in Equation~\eqref{eq:cm} and the minimal area
function $\mu_{\min}(b)$ in Section~\ref{sec:reduction}; and the letter $q$ denotes the
terminal point of a path polytope in Equation~\eqref{eq:polytope}, for which reason it has
been decorated as $\vq$ above, and the Tsallis index everywhere else. In each case the ambient
statement determines the reading, and no formula below contains both.

\subsection{The Length Measure and its Reflection}\label{sec:reflection}
Theorem C is stated as an identity between the propagator and a product of first order
operators applied to the continuum multinomial. The purpose of this subsection is to record
what that product of operators does, and thereby to fix the vocabulary used for the rest of
the paper. It is not the marginal spherical measure whose transform is the Klein--Gordon
propagator, but that measure differentiated $d-2$ times in the length variable; we call the
latter object the \textbf{derivated marginal spherical distribution}, and it is this object,
and not the marginal spherical distribution itself, which corresponds to the propagator and
which appears as the arm law throughout the paper.

For $\nu\in\R$ write
\begin{equation}\label{eq:rho}
\rho^{\nu}_{L}(\vx)=(L^{2}-r^{2})^{\nu-\frac12}_{+},\qquad r=|\vx|_{2}.
\end{equation}
The family of Equation~\eqref{eq:rho} is the family of projections of the uniform measure on a
sphere: the image of $\mathrm{Unif}(S^{n-1}(L))\subset\R^{n}$ under orthogonal projection onto
$\R^{k}$ has density proportional to $\rho^{\nu}_{L}$ with $\nu=\frac{c-1}{2}$, where $c=n-k$
is the codimension of the projection; projecting the uniform measure on the ball instead
shifts $c$ by two. The marginal spherical distribution of \cite{ODwyerQFT} is the member
$\nu=\frac{d}{2}-1$, that is, the codimension $d-1$ projection of a uniform sphere, or
equivalently the codimension $d-3$ projection of a uniform ball; at $d=3$ the latter
codimension is zero and the measure is exactly uniform on the ball.

\begin{lemma}[Basset]\label{lem:basset}
Let $r,m>0$ and $\nu>-\frac12$. Then
\[
\Lap|^{m}_{L}\big(\rho^{\nu}_{L}(\vx)\big)=\int_{r}^{\infty}e^{-mL}(L^{2}-r^{2})^{\nu-\frac12}dL
=\frac{\Gamma(\nu+\frac12)}{\sqrt{\pi}}2^{\nu}\Big(\frac{r}{m}\Big)^{\nu}K_{\nu}(mr),
\]
and both sides continue analytically in $\nu$.
\end{lemma}
\begin{proof}
Substitute $L=rt$ and use the Basset representation
$K_{\nu}(z)=\frac{\sqrt{\pi}(z/2)^{\nu}}{\Gamma(\nu+\frac12)}\int_{1}^{\infty}e^{-zt}(t^{2}-1)^{\nu-\frac12}dt$
\cite[3.387.3]{GR}.
\end{proof}

Comparing Lemma~\ref{lem:basset} at $\nu=\frac{d}{2}-1$ with the expression
$G(m,\vx)=(2\pi)^{-d/2}(m/r)^{\nu}K_{\nu}(mr)$ quoted in Theorem C, the Macdonald index is
already correct and only the prefactor is inverted, $(r/m)^{\nu}$ against $(m/r)^{\nu}$. This
is the whole of the discrepancy, and it is invisible to the Macdonald function because
$K_{-\nu}=K_{\nu}$.

\begin{proposition}\label{prop:reflection}
Let $d\ge3$ and $\nu=\frac{d}{2}-1$. Then, with $\partial^{d-2}_{L}$ taken in the sense of
distributions on $L\in(0,\infty)$,
\begin{equation}\label{eq:prop1}
G(m,\vx)=c_{d}\,r^{-(d-2)}\,\Lap|^{m}_{L}\Big[\partial^{d-2}_{L}\rho^{\frac{d}{2}-1}_{L}(\vx)\Big],
\qquad c_{d}=\frac{(2\pi)^{-\frac{d}{2}}\sqrt{\pi}}{\Gamma(\frac{d-1}{2})2^{\frac{d}{2}-1}}.
\end{equation}
Equivalently, $G$ is, up to the constant $c_{d}$ and the factor $r^{-(d-2)}$, the Laplace
transform of the member of the family of Equation~\eqref{eq:rho} obtained from the marginal
spherical member by the reflection
\[
\nu\longmapsto-\nu,\qquad\text{equivalently}\qquad c\longmapsto2-c,
\]
so that the marginal spherical codimension $c=d-1$ is carried to $c=3-d$.
\end{proposition}
\begin{proof}
The function $L\mapsto\rho^{\frac{d}{2}-1}_{L}(\vx)$ vanishes identically for $L<r$, so all of
its one sided derivatives at $L=0$ vanish and
$\Lap|^{m}_{L}[\partial^{n}_{L}\rho]=m^{n}\Lap|^{m}_{L}[\rho]$ with no boundary terms, for
every $n$, the derivative being distributional. Lemma~\ref{lem:basset} at $\nu=\frac{d}{2}-1$
therefore gives, since $2\nu=d-2$,
\[
\Lap|^{m}_{L}\big[\partial^{d-2}_{L}\rho^{\frac{d}{2}-1}\big]
=m^{2\nu}\cdot\frac{\Gamma(\nu+\frac12)}{\sqrt{\pi}}2^{\nu}r^{\nu}m^{-\nu}K_{\nu}(mr)
=\frac{\Gamma(\frac{d-1}{2})2^{\frac{d}{2}-1}}{\sqrt{\pi}}r^{d-2}\Big(\frac{m}{r}\Big)^{\nu}K_{\nu}(mr).
\]
Multiplying by $c_{d}r^{-(d-2)}$ gives $(2\pi)^{-d/2}(m/r)^{\nu}K_{\nu}(mr)=G(m,\vx)$. For the
second statement, the reflection $\nu\mapsto-\nu$ replaces
$\Gamma(\nu+\frac12)2^{\nu}(r/m)^{\nu}K_{\nu}$ by a constant multiple of $(m/r)^{\nu}K_{\nu}$,
which is the propagator; in terms of the codimension $c=2\nu+1$ this reads $c\mapsto2-c$.
\end{proof}

\begin{remark}\label{rem:branch}
Two features of Proposition~\ref{prop:reflection} are worth stating, because they are what
forces the language of the following sections.

First, the reflected member is a distribution and not a density. Its exponent
$-\nu-\frac12=-\frac{d-1}{2}$ is negative, so it fails to be locally integrable near $r=L$ once
$d\ge3$, and the constant $\Gamma(\frac{3-d}{2})$ produced by continuing
Lemma~\ref{lem:basset} has a pole at every odd $d$. It is for this reason that the derivated
marginal spherical distribution is a distribution rather than a density, and why the correct
statement is one about derivatives of the marginal spherical measure rather than about some
other measure. At $d=3$ the pole is the simplest possible one and Equation~\eqref{eq:prop1}
degenerates to
\[
G(m,\vx)=\frac{1}{4\pi r}\Lap|^{m}_{L}\big[\delta(L-r)\big]=\frac{e^{-mr}}{4\pi r},
\]
i.e.\ the length measure of the propagator is carried entirely by straight line paths, and the
reflected codimension $3-d$ is zero.

Second, this is precisely the structure of Theorem C. There the propagator is not the
transform of the continuum multinomial but the transform of the continuum multinomial acted on
by the product $\prod_{j=1}^{p}(m(p-1+\frac{1}{dp})+\va_{j}*\nabla)$ of $p$ first order
operators; Proposition~\ref{prop:reflection} is the statement of what that product does in the
limit, namely that it supplies the $d-2$ derivatives which carry $\nu$ to $-\nu$. Wherever the
phrase ``the one dimensional Polyakov measure of length $L$'' is used below it refers to the
derivated marginal spherical distribution, that is, to $\rho^{\frac{d}{2}-1}_{L}$ together with
the $d-2$ derivatives carried outside, exactly as they are in Equation~\eqref{eq:thmC}.
\end{remark}

The following observation is the reason the two-dimensional theory of this paper is the
correct generalisation of the one dimensional theory of \cite{ODwyerQFT}

\begin{proposition}[The two branches of the area law]\label{prop:branches}
Let $\alpha>0$ and $\gamma>0$, and set
\begin{equation}\label{eq:compact}
f^{\gamma}_{\alpha}(A)=\frac{A}{(\alpha^{2}+A^{2})^{\gamma}},\qquad
g^{\gamma}_{\alpha}(A)=\frac{A}{(\alpha^{2}-A^{2})^{\gamma}}\mathbf{1}_{0<A<\alpha}.
\end{equation}
Then:
\begin{enumerate}
\item $f^{\gamma}_{\alpha}$ and $g^{\gamma}_{\alpha}$ are carried into one another by the
substitution $\alpha^{2}\mapsto-\alpha^{2}$, equivalently $A\mapsto iA$, up to an overall constant,
and this holds at every $\gamma$ and not only at $\gamma=2$;
\item $g^{\gamma}_{\alpha}$ is supported in $|A|<\alpha$ and is not integrable there for
$\gamma\ge1$, the exponent producing a divergence at the endpoint $|A|=\alpha$;
\item by Proposition~\ref{prop:reflection} the derivated marginal spherical distribution in ambient
dimension $d$ has bracket exponent $-\frac{d-1}{2}$, so at $d=4$ its compact form is
$g^{\frac32}_{\alpha}$ and its heavy tailed partner is $f^{\frac32}_{\alpha}$, whose survival
function is $(1+(A/\alpha)^{2})^{-\frac12}$.
\end{enumerate}
\end{proposition}

Part (3) concerns the one-dimensional reflected member itself. At $d=4$ that member has
$\gamma=\frac32$, whereas the transverse two-dimensional reading used to form a triangle has a
different exponent for the ordinary reason encoded in Lemma~\ref{lem:mshift}: Equation~\eqref{eq:mvalue}
gives $m=d-2$, and the two-dimensional reading therefore has
$\gamma=(m+2)/2=d/2$, equal to $2$ at $d=4$. This is a change of marginal dimension, not an
additional dynamical correction. In particular no surplus parameter, extra ambient dimensions, or
additional surface derivative is introduced to obtain the Ansoldi exponent.

Under the fixed-side configuration of Proposition~\ref{prop:triangle},
$A=\frac{\ell}{2}|\vW|$, so the area law is the two-dimensional arm law rescaled. This elementary
linearity is precisely the desired mechanism: with $d=D$ it transfers the exponent $d/2$ directly to
the area and hence reproduces the arbitrary-dimensional Ansoldi family. The compact and heavy-tailed
branches of Equation~\eqref{eq:compact} retain the same $\gamma$; the branch substitution changes the
sign of the quadratic, while Lemma~\ref{lem:mshift} accounts for the one- versus two-dimensional
reading. The $d-2$ derivatives of Proposition~\ref{prop:reflection} remain part of the one-dimensional
propagator statement and stand outside the surface pushforward.

\subsection{From Paths to Surfaces}\label{sec:dictionary}
The measure quoted in Theorem C counts directed paths between two points, and is transformed
from the length of those paths to the mass as its conjugate variable. Everything this paper
does is the same construction one dimension up: a measure counting surfaces between two
contours, transformed from the area of those surfaces to the brane tension as its conjugate
variable. The dictionary is the following one.

\begin{center}
\begin{tabular}{ll}
\hline
one dimensional & two dimensional\\
\hline
directed path between two points & worldsheet $K\in\Gamma_{\sigma_{i},\sigma_{j}}$ between two contours\\
length $L$ & area $A$\\
mass $m$ & brane tension $T$\\
$\Lap|^{m}_{L}$ & $\Four|^{T}_{A}$\\
displacement $r=|\vx|_{2}$ & minimal area $\mu_{\min}$; $=\sqrt{b/2}$ when flat\\
derivated marginal spherical $\partial^{d-2}_{L}\rho^{\frac{d}{2}-1}_{L}$ & Fisk density $f_{\alpha}$ of Equation~\eqref{eq:fisk}\\
$G(m,\vx)$ & $K_{l^{*}_{2}}(\sigma_{i},\sigma_{j})$\\
\hline
\end{tabular}
\end{center}

It is worth setting the measure down explicitly here, since Sections~\ref{sec:triangle}
and~\ref{sec:reduction} compute with it and the numerical experiment of
Section~\ref{sec:numerics} samples from it. Let $\sigma_{i}$ and $\sigma_{j}$ be two disjoint
boundary contours and let $\Gamma_{\sigma_{i},\sigma_{j}}$ be the set of polyhedral surfaces
interpolating them. A given $K\in\Gamma_{\sigma_{i},\sigma_{j}}$ carries a skeleton of $M$ one
dimensional arms, or spines, joining $\sigma_{i}$ to $\sigma_{j}$; the measure assigned to
$K$ is the product over those $M$ spines of the marginal spherical measure
$\rho^{\frac{d}{2}-1}_{L}$ of Equation~\eqref{eq:rho}, and the $d-2$ derivatives of
Proposition~\ref{prop:reflection} are applied \emph{after} the pushforward, in the length
variable, exactly as the product $\prod_{j}(ms_{p}+\va_{j}*\nabla)$ stands outside the
continuum multinomial in Theorem C. The faces of $K$ are the two simplices spanned by
consecutive spines and $\mu(K)$ is the sum of their areas, so that $\mu$ is the Nambu--Goto area
of the complex. The extended body propagator $K_{l^{*}_{2}}(\sigma_{i},\sigma_{j})$ is then the
Fourier transform, from $\mu(K)$ to the brane tension $T$, of the pushforward of that product
measure under $K\mapsto\mu(K)$, followed by the derivatives just described. Written in one
line for $p$ dimensional bodies rather than two, the same prescription reads: draw the spines
of the skeleton from the marginal spherical measure, take the $p$ volume of a $p$ simplex to be
the Gram determinant $\Vol_{p}(v_{1},\dots,v_{p})=\frac{1}{p!}\sqrt{\det(v_{i}\cdot v_{j})}
=\frac{1}{p!}\|v_{1}\wedge\cdots\wedge v_{p}\|$ of the $p$ spines spanning it, and transform
the sum of those volumes against the brane tension; at $p=1$ this returns the construction of
\cite{ODwyerQFT} and at $p=2$ it returns the measure of this paper, so that the extension to
branes is immediate.

The last clause is worth stating on its own, since it is the form in which the construction is
conjectured to extend beyond surfaces. Let $p\ge1$ and let two $p-1$ dimensional boundary data
be given. The conjecture is that the measure over $p$ dimensional simplicial complexes
interpolating them is obtained by exactly the prescription above and by no other ingredient: the
skeleton of the complex carries one dimensional spines, each drawn from the marginal spherical
measure; the $p$ volume of a $p$ simplex of the complex is the Gram determinant of the $p$
spines which span it, so that the sum of those volumes is the Nambu--Goto volume of the complex,
which is the simplicial extension of the Nambu--Goto action; and the $p$ brane propagator is the
transform of the pushforward of the spine measure under that volume, against the brane tension as
conjugate variable. At $p=1$ the Gram determinant is the length of a segment and the prescription
returns the construction of \cite{ODwyerQFT}; at $p=2$ it is the area of a triangle and the
prescription returns the measure of this paper. What is being conjectured is only that no new
ingredient is required at higher $p$, namely the same spines, the same product measure over the
skeleton, and the Gram determinant as the volume element. The law of that volume for $p\ge3$,
which the author intends to treat separately, is conjectured to be the law of the Gram
determinant of a simplex spanned by $p$ many marginal spherical arms (constrained by the boundary
data).

Two features of this definition should be recorded before it is used. The first is that for any
finite number $M$ of spines the measure is well defined without further argument. Each factor is
$\rho^{\frac{d}{2}-1}_{L}$, which by Equation~\eqref{eq:rho} is a finite positive measure with
compact support in $\vx$ for each fixed $L$, so the product over the $M$ spines is the ordinary
product measure on a product of compact sets; the area $\mu(K)$ is a continuous function of the
spines, so the pushforward is a finite positive measure on $[0,\infty)$; and the $d-2$
derivatives of Proposition~\ref{prop:reflection} are applied to that pushforward, in the length
variable, in the sense of distributions. This ordering matters. The derivated object of
Remark~\ref{rem:branch} is a distribution and not a density: it is not locally integrable, it is
not positive, and the phrase ``draw a spine from it'' has no meaning.

The second is that the content of the conjecture lies entirely in the limit $M\to\infty$, and the
author expects that limit to be easier to establish through the area law than directly. That is,
rather than construct the limiting measure and then compute its area law, one shows that the area
law at finite but large $M$ is approximately the Fisk law of Equation~\eqref{eq:fisk}, with an
error controlled in $M$; the existence of the limit then follows from the convergence of those
approximations rather than being assumed in advance. Section~\ref{sec:numerics} is the numerical
form of exactly this statement, and it is why the experiment reported there is organised around
fidelity to the Fisk law at a sequence of finite $M$ rather than around any attempted construction
of the limit.

The right-hand column is what Sections~\ref{sec:ansoldi} to~\ref{sec:triangle} establish.

One asymmetry in the dictionary should be stated rather than smoothed over. In the one dimensional
theory the derivatives of Proposition~\ref{prop:reflection} are genuinely required: the marginal
spherical measure transforms to $(r/m)^{\nu}K_{\nu}$ and not to the propagator, and it is the
derivated marginal spherical distribution which does. In the two dimensional theory no such
derivatives appear, and Theorem~\ref{thm:ansoldi} exhibits the loop space propagator as the direct
transform of the Fisk density. By Proposition~\ref{prop:branches} the two statements are the same
statement read on the two branches, so the asymmetry is one of presentation rather than of content.
The author's expectation is that the two dimensional analogue of those derivatives has already been
absorbed into the Schwinger parameter integral over $W$ in Equation~\eqref{eq:ansoldi}, which has no
one dimensional counterpart in the form quoted in Theorem C. This is a conjecture and is not used
anywhere below.

\section{The Fisk Law, \texorpdfstring{$q$}{q}-Tsallis Statistics, and Two Transform Identities}\label{sec:fisklaw}

The final set of definitions required for this paper surround $q$-Tsallis statistics and a number
of concepts related to maximal entropy distributions. They are quoted here only so that the density
of Equation~\eqref{eq:fisk} may be named in the language of \cite{ODwyerQFT}; nothing in
Sections~\ref{sec:ansoldi}--\ref{sec:jt} depends on them. For a random variable with distribution
$p$, its $q$-Tsallis entropy is given in Equation~\eqref{eq:tsallis}, where
$\ln_{q}(x)=\frac{x^{1-q}-1}{1-q}$ is the $q$-logarithm \cite{AbulMagd}.
\begin{equation}\label{eq:tsallis}
S_{q}=-\int p^{q}\ln_{q}(p)\ =\ \frac{1-\int p^{q}}{q-1}
\end{equation}

The sign is the one which makes $S_{q}$ an entropy: with $\ln_{q}$ as defined above,
$\int p^{q}\ln_{q}p\to\int p\ln p$ as $q\to1$, which is minus the Shannon entropy. Thus the
entropy convention used here is $S_{q}=-\int p^{q}\ln_{q}p$, and
$S_{q}\to-\int p\ln p$ as $q\to1$. Like the Shannon entropy, this functional over random variable distributions has a unique maximal
entropy distribution given some constraints. These become the $q$ variants of the normal maximum
entropy distribution; the maximum entropy distribution for the $q$-Tsallis entropy under
normalisation together with a constraint on the second moment (equivalently, on the $q$-escort
variance) is the $q$-Gaussian, or $\mathrm{gaus}_{q}$, distribution. This is given in
Equation~\eqref{eq:gausq}, where $C_{q}$ is determined by normalization and $\beta^{*}_{q}$ is the
coefficient elaborated upon in \cite{Suyari}.
\begin{equation}\label{eq:gausq}
\mathrm{gaus}_{q}(x)=C_{q}\big(1+(q-1)\beta^{*}_{q}x^{2}\big)^{-\frac{1}{q-1}}
\end{equation}
For particular values of $q$, Equation~\eqref{eq:gausq} becomes the marginal spherical distribution
of \cite{ODwyerQFT} or the multivariate Student's $t$ distribution.

The index $q$ is not necessarily dimensionally invariant, and since the construction below reads one
and the same object first in one dimension and then in two, what is invariant must be named. Write an
isotropic member of Equation~\eqref{eq:gausq} in $\R^{k}$ as
$\propto(1+\beta|x|^{2})^{-\frac{m+k}{2}}$, so that $\frac{1}{q-1}=\frac{m+k}{2}$, that is
$q_{k}=1+\frac{2}{m+k}$.

\begin{lemma}[The invariant of the family]\label{lem:mshift}
Let $m>0$ and let $X$ be isotropic in $\R^{k}$ with density
$\propto(1+\beta|x|^{2})^{-\frac{m+k}{2}}$. Then every $j$ dimensional marginal of $X$ has density
$\propto(1+\beta|y|^{2})^{-\frac{m+j}{2}}$, with the same $m$ and $\beta$. Consequently $m$ is
independent of the dimension in which the density is read, while $q_{k}=1+\frac{2}{m+k}$ is not.
\end{lemma}
\begin{proof}
This is the standard multivariate Student $t$ with $m$ degrees of freedom and scale
$\beta^{-1/2}$, whose marginals are Student $t$ with the same $m$; equivalently, writing
$X=\sqrt{m/S}\,\beta^{-1/2}G$ with $G\sim N(0,I_{k})$ and $S\sim\chi^{2}_{m}$ independent, any
$j$ coordinates of $G$ are $N(0,I_{j})$ and the scalar factor is untouched.
\end{proof}

Two readings of one object are used below and Lemma~\ref{lem:mshift} is what carries between
them. The reflected member of Proposition~\ref{prop:reflection} has exponent $-\frac{d-1}{2}$ in the
length variable, which is one dimensional; comparing with $-\frac{m+1}{2}$ gives
\begin{equation}\label{eq:mvalue}
m=d-2
\end{equation}
Section~\ref{sec:triangle} reads the same system in the two plane transverse to a fixed base, where
Lemma~\ref{lem:mshift} gives exponent $-\frac{m+2}{2}=-\frac{d}{2}$ and hence
$q_{2}=1+\frac{2}{d}$. It is this second reading, and not the first, which supplies the index used in
Proposition~\ref{prop:triangle}. No such correction is
required for the Ansoldi comparison: the two-dimensional exponent is $d/2$, exactly
Equation~\eqref{eq:gammaD} when $d=D$.

The object which the surface measure of Section~\ref{sec:dictionary} actually produces is not a
$q$-Gaussian but the volume of a simplex spanned by $q$-Gaussian arms, and at $k=2$ that volume is
the Gram determinant of two vectors. We record the law of that determinant here, since
Section~\ref{sec:triangle} computes with it. Write
\begin{equation}\label{eq:qchi}
\chi^{q}_{k}(x)\sim x^{k-1}\big(1+(q-1)\beta^{*}_{q}x^{2}\big)^{-\frac{1}{q-1}},
\end{equation}
and call $\chi^{q}_{2}$ the $q$-Tsallis chi distribution with two degrees of freedom.

\begin{lemma}[The determinant law at $k=2$]\label{lem:det}
Let $1<q<2$ and let $\vW$ be a random vector in $\R^{2}$ whose law is the two dimensional
$q$-Gaussian of Equation~\eqref{eq:gausq}, that is, with density
$C_{q}(1+(q-1)\beta^{*}_{q}|\vec{w}|^{2})^{-\frac{1}{q-1}}$. Let $v_{1},v_{2}\in\R^{d}$ span a
$2$-simplex whose two spines differ by a vector $v_{2}-v_{1}=\ell e$ of fixed length $\ell$, and let
$\vW$ be the component of $v_{1}$ in the plane orthogonal to $e$. Then
\begin{enumerate}
\item $|\vW|$ has density $\chi^{q}_{2}$;
\item the Gram determinant of the simplex,
$\Vol_{2}(v_{1},v_{2})=\frac12\|v_{1}\wedge v_{2}\|$, equals $\frac{\ell}{2}|\vW|$ and therefore has
the law $\chi^{q}_{2}$ with $\beta^{*}_{q}$ replaced by $(2/\ell)^{2}\beta^{*}_{q}$.
\end{enumerate}
\end{lemma}
\begin{proof}
For (1), the law of $\vW$ is rotationally symmetric, so in polar coordinates
$\mathbb{P}(|\vW|\in dr)=2\pi C_{q}\,r(1+(q-1)\beta^{*}_{q}r^{2})^{-\frac{1}{q-1}}dr$, which is
Equation~\eqref{eq:qchi} at $k=2$. For (2), $v_{2}-v_{1}$ is a linear combination of $v_{1}$ and $v_{2}$, so
$\|v_{1}\wedge v_{2}\|=\|v_{1}\wedge(v_{2}-v_{1})\|=\ell\,\|v_{1}\wedge e\|=\ell\,|\vW|$, the last
equality because $\|v_{1}\wedge e\|$ is the length of the component of $v_{1}$ orthogonal to the unit
vector $e$. The family of Equation~\eqref{eq:qchi} at $k=2$ is closed under positive scaling, with the
scale parameter transforming as stated.
\end{proof}

Only $k=2$ is used in this paper, and Lemma~\ref{lem:det} is the whole of what is needed. The
corresponding statement at $k\ge3$ is not obtained by taking $k$ independent copies of
Equation~\eqref{eq:gausq}: a product of independent $q$-Gaussians is not rotationally symmetric, and
at fixed $q=\frac32$ the density of Equation~\eqref{eq:qchi} behaves as $x^{k-5}$ at infinity and so
fails to normalise once $k\ge4$, the exponent having to be allowed to depend on $k$. The extension of
the construction to $p\ge3$ simplices we will conjecture in Section~\ref{sec:dictionary} requires the law of a
$p$-fold Gram determinant rather than a $\chi^{q}_{k}$.

With these conventions fixed we may name Equation~\eqref{eq:fisk}.

\begin{definition}\label{def:fisk}
For $\alpha>0$ the Fisk density of scale $\alpha$ is
\[
f_{\alpha}(A)=\frac{2\alpha^{2}A}{(\alpha^{2}+A^{2})^{2}}\mathbf{1}_{A>0},\qquad
\int_{0}^{\infty}f_{\alpha}=1,\qquad \mathbb{P}(A>a)=\frac{1}{1+(a/\alpha)^{2}}.
\]
We write $f^{\mathrm{odd}}_{\alpha}$ for its odd extension to $\R$. This is
Equation~\eqref{eq:fisk} with the normalisation supplied.
\end{definition}

The following two elementary identities carry the whole of
Sections~\ref{sec:ansoldi}--\ref{sec:jt}. The first is the Schwinger parameter integral in the only
exponent at which the Macdonald function degenerates to an elementary one; the second is the
transform pair which turns that elementary answer back into a density on areas.

\begin{lemma}[Schwinger--Macdonald]\label{lem:schwinger}
Let $\theta\in\R$ and let $p,q\in\mathbb{C}$ with $\mathrm{Re}\,p>0$ and $\mathrm{Re}\,q>0$. Then
\[
\int_{0}^{\infty}W^{-\theta}e^{-\frac{p}{W}-qW}dW=2\Big(\frac{p}{q}\Big)^{\frac{1-\theta}{2}}
K_{1-\theta}\big(2\sqrt{pq}\big),
\]
where $K$ is the Macdonald function and the branches are the principal ones. Both sides are analytic
in $(p,q)$ and the identity persists on the boundary $\mathrm{Re}\,p=0$ or $\mathrm{Re}\,q=0$ by
continuation, provided $\mathrm{Re}\sqrt{pq}>0$. For $\theta=\frac32$ the right hand side is
elementary:
\[
\int_{0}^{\infty}W^{-\frac32}e^{-\frac{p}{W}-qW}dW=\sqrt{\frac{\pi}{p}}\,e^{-2\sqrt{pq}}.
\]
\end{lemma}
\begin{proof}
The first identity is \cite[3.471.9]{GR}. For the second, $K_{-\frac12}=K_{\frac12}$ and
$K_{\frac12}(z)=\sqrt{\pi/2z}\,e^{-z}$, so that
$2(p/q)^{-\frac14}K_{-\frac12}(2\sqrt{pq})=2(q/p)^{\frac14}\frac12\sqrt{\pi}(pq)^{-\frac14}e^{-2\sqrt{pq}}
=\sqrt{\pi/p}\,e^{-2\sqrt{pq}}$.
\end{proof}

\begin{lemma}[Fisk sine transform]\label{lem:sine}
Let $\alpha>0$. Then for every $T>0$
\[
\int_{0}^{\infty}f_{\alpha}(A)\sin(TA)\,dA=\frac{\pi\alpha}{2}Te^{-\alpha T},
\qquad\text{equivalently}\qquad
\Four|^{T}_{A}\big(f^{\mathrm{odd}}_{\alpha}\big)=i\pi\alpha Te^{-\alpha|T|}.
\]
More generally, for $\gamma>\frac12$,
\[
\int_{0}^{\infty}\frac{A\sin(TA)}{(\alpha^{2}+A^{2})^{\gamma}}dA
=\frac{\sqrt{\pi}}{2^{\gamma-\frac12}\Gamma(\gamma)}\alpha^{\frac32-\gamma}T^{\gamma-\frac12}
K_{\gamma-\frac32}(\alpha T).
\]
\end{lemma}
\begin{proof}
Write $g(A)=(\alpha^{2}+A^{2})^{-1}$, whose Fourier transform is
$\Four|^{T}_{A}(g)=\frac{\pi}{\alpha}e^{-\alpha|T|}$. Since
$A(\alpha^{2}+A^{2})^{-2}=-\frac12g'(A)$ and $\Four(g')=-iT\Four(g)$, the odd function
$2\alpha^{2}A(\alpha^{2}+A^{2})^{-2}$ has Fourier transform
$\alpha^{2}iT\frac{\pi}{\alpha}e^{-\alpha|T|}=i\pi\alpha Te^{-\alpha|T|}$; for an odd function
$\Four|^{T}_{A}=2i\int_{0}^{\infty}(\cdot)\sin(TA)dA$, which gives the first display. The general
$\gamma$ statement is the Fourier sine transform tabulated as entry 11 of \cite[12.33]{GR}, in the
form $x(x^{2}+a^{2})^{-\nu-\frac32}\mapsto\xi^{\nu+1}\sqrt{2}(2a)^{-\nu}\Gamma(\nu+\frac32)^{-1}
K_{\nu}(a\xi)$; putting $\gamma=\nu+\frac32$ and undoing the $\sqrt{2/\pi}$ normalisation of the
transform used there reproduces the display above, and it reduces to the first display at
$\gamma=2$ by $K_{\frac12}(z)=\sqrt{\pi/2z}e^{-z}$.
\end{proof}

Lemma~\ref{lem:sine} is the reason the exponent $\frac32$ in the first line of
Equation~\eqref{eq:three} is not an accident. By Lemma~\ref{lem:schwinger}, a Schwinger integral with
exponent $\theta$ returns $K_{1-\theta}$, and by Lemma~\ref{lem:sine} a density
$A(\alpha^{2}+A^{2})^{-\gamma}$ returns $K_{\gamma-\frac32}$; the two match when $|1-\theta|=|\gamma-\frac32|$, since $K_{-\nu}=K_{\nu}$. Matching the powers of $T$ selects $\gamma=\theta+\frac12$: the second Bessel-order root $\gamma=\frac52-\theta$ gives $T^{2-\theta}$ and agrees with the Schwinger power $T^{\theta}$ only at $D=3$. The Schwinger exponent of Equation~\eqref{eq:ansoldi} counts the momenta
transverse to a loop in the ambient dimension $D$\footnote{In \cite[Eqns.~(5.9) and (5.18)]{Ansoldi}, $\delta\Sigma^{\mu\nu}/\delta x^{\alpha}=\delta^{\mu}_{\alpha}x'^{\nu}-\delta^{\nu}_{\alpha}x'^{\mu}$; contracting its square gives $2(D-1)x'^2$, while $p_{\mu}=p_{\mu\nu}x'^{\nu}$ is orthogonal to $x'^{\mu}$ and hence has $D-1$ normal components.}, so that $\theta=\frac{D-1}{2}$ and
\begin{equation}\label{eq:gammaD}
\gamma=\frac{D}{2},\qquad\text{equivalently}\qquad q=1+\frac{2}{D},
\end{equation}
the Macdonald orders being $K_{\frac{3-D}{2}}$ and $K_{\frac{D-3}{2}}$ and equal for every $D$.
At $D=4$ this is $\theta=\frac32$ and $\gamma=2$. The whole one
parameter family $A(\alpha^{2}+A^{2})^{-\gamma}$ is the family $\chi^{q}_{2}$, and the Fisk density is
the single member of it at which both transforms degenerate to elementary functions, the
Macdonald order being $\pm\frac12$ there and nowhere else.

Equation~\eqref{eq:gammaD} is the arbitrary-dimensional exponent of the Ansoldi
loop-space kernel. Equation~\eqref{eq:mvalue} together with Lemma~\ref{lem:mshift} is the exponent
the fixed-side geometry of Section~\ref{sec:triangle} carries. They agree identically under the natural identification $d=D$: the one-dimensional invariant
$m=d-2$ becomes the two-dimensional exponent $(m+2)/2=d/2=D/2$. This dimension-for-dimension
identity is the comparison used throughout the geometric argument.

\begin{proposition}\label{prop:index}
Let $q>1$ and let $\chi^{q}_{2}$ be as in Equation~\eqref{eq:qchi} with $k=2$. Then
$\chi^{q}_{2}\propto x(\alpha^{2}+x^{2})^{-\gamma}$ with $\gamma=\frac{1}{q-1}$, so that
Equation~\eqref{eq:gammaD} selects $q=1+\frac{2}{D}$ in ambient dimension $D$. Then $\chi^{q}_{2}$ is
a Fisk density if and only if $q=\frac32$, that is if and only if $D=4$, in which case
\[
\chi^{3/2}_{2}(x)\propto\frac{x}{\big(1+\frac12\beta^{*}x^{2}\big)^{2}}
=\frac{2}{\beta^{*2}}\cdot\frac{x}{\big(\frac{2}{\beta^{*}}+x^{2}\big)^{2}},
\]
i.e.\ $\chi^{3/2}_{2}=f_{\alpha}$ with $\alpha=\sqrt{2/\beta^{*}}$. Equivalently: a positive random
variable $A$ has the Fisk law of scale $\alpha$ if and only if $A^{2}/\alpha^{2}$ has the
Fisher--Snedecor $F(2,2)$ law, that is, if and only if $A^{2}/\alpha^{2}$ is the ratio of two
independent unit exponentials.
\end{proposition}
\begin{proof}
By Equation~\eqref{eq:qchi} with $k=2$,
$\chi^{q}_{2}(x)\propto x(1+(q-1)\beta^{*}x^{2})^{-\frac{1}{q-1}}$. A Fisk density is
$\propto x(\alpha^{2}+x^{2})^{-2}$, so the exponent must satisfy $\frac{1}{q-1}=2$, giving
$q=\frac32$; conversely $q=\frac32$ gives the displayed form. For the second statement, if
$A\sim f_{\alpha}$ then $\mathbb{P}(A^{2}/\alpha^{2}>y)=\mathbb{P}(A>\alpha\sqrt{y})=(1+y)^{-1}$,
whose density $(1+y)^{-2}$ is that of $F(2,2)$, equivalently of $\beta'(1,1)$, equivalently of $U/V$
for $U,V$ independent unit exponentials; and conversely, since these determine the law.
\end{proof}

\section{The Propagator of Ansoldi et al.}\label{sec:ansoldi}

Using a Hamilton--Jacobi formulation of the relativistic quantum mechanics of extended objects,
Ansoldi et al.\ \cite{Ansoldi} devised an equation for the free scalar propagator of strings in loop
space. Their result is written in Equation~\eqref{eq:ansoldi}. Ansoldi et al.\ used $A$ instead of
$W$ in their work, but we will adopt $W$ above as we would like to reserve $A$ to denote a genuine
area of some worldsheet.
\begin{equation}\label{eq:ansoldi}
\int_{0}^{\infty}dW\,e^{-\frac{im^{2}W}{2}}\Big(\frac{m^{2}}{2i\pi W}\Big)^{\frac32}
e^{\frac{im^{2}b}{4W}}
\end{equation}
In calculations of point correlators in String Theory, by using the variational principle on the
action of the string along with a mode expansion on the Fourier modes of the string, one can relate
the mass squared of the string to its tension times some factor related to the strings' excitatory
modes. The contribution from excitatory modes is tension independent in this derivation, and
therefore it is a valid assumption to relate the $m^{2}$ in Equation~\eqref{eq:ansoldi} to $T$ by a
linear correspondence. Since this paper is concerned with deriving the propagators of physics from
genuine, well-defined measures on spaces of manifolds, this choice may as well be arbitrary, though
motivated by physics, so long as we arrive at the correct propagator from our measure. If
$m^{2}=sT$, where $s$ is some unconstrained constant and $T$ is the string tension,
Equation~\eqref{eq:ansoldi} becomes Equation~\eqref{eq:ansoldiT}.
\begin{equation}\label{eq:ansoldiT}
\int_{0}^{\infty}dW\Big(\frac{sT}{2i\pi W}\Big)^{\frac32}e^{isT\left(\frac{b}{4W}-\frac{W}{2}\right)}
\end{equation}
Since the simplicial complex propagator for strings is meant to be a Fourier transform of the measure
of worldsheets of area $A$ to a parameter $T$, its inverse transform should return the measure over
worldsheets.

\begin{theorem}\label{thm:ansoldi}
Let $b>0$ and let $G(m^{2})$ denote the loop space propagator of Equation~\eqref{eq:ansoldi}, the
integral being understood as the limit as $\epsilon\downarrow0$ of the same integrand with
$e^{-\epsilon(W+W^{-1})}$ inserted. Then
\begin{equation}\label{eq:thm1}
G(m^{2})=\frac{-i}{\pi\sqrt{2b}}m^{2}e^{-m^{2}\sqrt{b/2}}
=\frac{-2i}{\pi^{2}b}\int_{0}^{\infty}f_{\sqrt{b/2}}(A)\sin\big(m^{2}A\big)dA
=\frac{-1}{\pi^{2}b}\Four|^{m^{2}}_{A}\big(f^{\mathrm{odd}}_{\sqrt{b/2}}\big).
\end{equation}
In particular the measure over worldsheet areas whose Fourier transform is the Ansoldi loop space
propagator is exactly the Fisk density of scale $\alpha=\sqrt{b/2}$.
\end{theorem}
\begin{proof}
Set $M=m^{2}$. The integral in Equation~\eqref{eq:ansoldi} is
$(M/2i\pi)^{3/2}\int_{0}^{\infty}W^{-3/2}e^{-p/W-qW}dW$ with $p=-iMb/4$ and $q=iM/2$, regularised as
stated so that $\mathrm{Re}\,p,\mathrm{Re}\,q>0$. Then $pq=M^{2}b/8$ is positive real, so
$2\sqrt{pq}=M\sqrt{b/2}$ and $\mathrm{Re}\sqrt{pq}>0$. Lemma~\ref{lem:schwinger} at $\theta=\frac32$
gives, using $1/(-i)=i$,
\[
\int_{0}^{\infty}W^{-\frac32}e^{-\frac{p}{W}-qW}dW=\sqrt{\frac{\pi}{p}}e^{-M\sqrt{b/2}}
=\frac{2\sqrt{\pi i}}{M^{\frac12}\sqrt{b}}e^{-M\sqrt{b/2}}.
\]
Multiplying by $(M/2i\pi)^{3/2}=M^{3/2}(2\pi)^{-3/2}e^{-3i\pi/4}$, the powers of $M$ combine to
$M^{3/2}M^{-1/2}=M$ and the constants to
$2\sqrt{\pi}e^{i\pi/4}(2\pi)^{-3/2}e^{-3i\pi/4}=2\sqrt{\pi}(2\pi)^{-3/2}e^{-i\pi/2}=-i/(\sqrt{2}\pi)$.
This is the first equality of Equation~\eqref{eq:thm1}. The second is Lemma~\ref{lem:sine} with
$\alpha=\sqrt{b/2}$, which gives $\int_{0}^{\infty}f_{\alpha}(A)\sin(MA)dA=\frac{\pi\alpha}{2}Me^{-\alpha M}$,
together with $\frac{2}{\pi^{2}b}\cdot\frac{\pi\alpha}{2}=\frac{\alpha}{\pi b}=\frac{1}{\pi\sqrt{2b}}$.
The third is the same identity written for the odd extension, since
$\Four|^{T}_{A}(f^{\mathrm{odd}})=2i\int_{0}^{\infty}f\sin(TA)dA$.
\end{proof}

Written as a density in $A$, Theorem~\ref{thm:ansoldi} says that the answer is the second line of
Equation~\eqref{eq:fiskwritten}.
\begin{equation}\label{eq:fiskwritten}
f(A)=C(b,s)\frac{\big(\tfrac{A}{\sqrt{b}}\big)}{\Big(1+2\big(\tfrac{A}{\sqrt{b}}\big)^{2}\Big)^{2}}
=C'(b,s)\frac{A}{(b+2A^{2})^{2}}=C''(b,s)\,f_{\sqrt{b/2}}(A)
\end{equation}
This last line is $C(b,s)\chi^{\frac32}_{2}(A/\sqrt{b})$; normalization fixes the coefficient
$C(b,s)$ to its desired value in Section~\ref{sec:fisklaw}. Therefore its
equivalence with the loop space propagator of \cite{Ansoldi} is immediately proven.

\begin{corollary}[Arbitrary ambient dimension]\label{cor:generalD}
Let $D\ge3$ and let $G_{D}(m^{2})$ denote Equation~\eqref{eq:ansoldi} with the exponent
$\frac32$ replaced by $\frac{D-1}{2}$, which is the form the loop space construction takes when the
Gaussian integration runs over the $D-1$ momenta transverse to the loop. Then $G_{D}$ is, up to a
constant, the sine transform of $A(\alpha^{2}+A^{2})^{-D/2}$ with $\alpha=\sqrt{b/2}$, that is of the
member $q=1+\frac{2}{D}$ of Equation~\eqref{eq:qchi} at $k=2$.
\end{corollary}
\begin{proof}
Lemma~\ref{lem:schwinger} at $\theta=\frac{D-1}{2}$ returns $K_{\frac{3-D}{2}}(2\sqrt{pq})$
with $2\sqrt{pq}=m^{2}\sqrt{b/2}$ as in Theorem~\ref{thm:ansoldi}, and Lemma~\ref{lem:sine} at
$\gamma=\frac{D}{2}$ returns $K_{\frac{D-3}{2}}(\alpha m^{2})$. Since $K_{-\nu}=K_{\nu}$, both sides are proportional to $\alpha^{(3-D)/2}(m^{2})^{(D-1)/2}K_{|D-3|/2}(\alpha m^{2})$.
\end{proof}

Theorem~\ref{thm:ansoldi} is the case $D=4$ of Corollary~\ref{cor:generalD}, and is stated
separately only because it is the case in which both transforms are elementary. Nothing below depends
on $D=4$ except where it is said to.

\begin{remark}[In what sense the two theories agree]\label{rem:quench}
Theorem~\ref{thm:ansoldi} is a statement about the propagator of \cite{Ansoldi}, and its
scale $\sqrt{b/2}$ is theirs rather than that of the present construction. The relationship between
the two is worth stating plainly, since the boundary data are not the same objects.

The construction of \cite{Ansoldi} is not a theory of surfaces. A loop $\sigma$ is compressed
to its array of signed shadow areas $\Sigma^{\mu\nu}(\sigma)=\int_{\sigma}x^{\mu}dx^{\nu}$, and
the propagator is an ordinary worldline propagator between the two points
$\Sigma(\sigma_{i}),\Sigma(\sigma_{j})$ of that space; this is why $b$ sits in
Equation~\eqref{eq:ansoldi} in exactly the slot occupied by the squared displacement in the
Schwinger representation of a point particle. Their boundary datum is the difference of shadow
areas; here it is the pair of contours itself, and the map $K\mapsto\int_{K}dx^{\mu}\wedge dx^{\nu}$
carries the second to the first, forgetting everything except the signed shadow.

Both theories then answer the same question in the same way. Each records the boundary
data through the smallest area compatible with them. For \cite{Ansoldi} that
smallest area is realised by a flat patch and equals $\sqrt{b/2}$ exactly, by the equality case of
Proposition~\ref{prop:signedbound}; for the measure of Section~\ref{sec:dictionary} it is the area
$\mu_{\min}$ of the Plateau solution. The two coincide precisely when the extremal surface is flat. The identification is to be read on the compact branch, for the reason recorded in
Remark~\ref{rem:branchcare}. On the Fisk branch $\alpha$ is the median of a density supported down to
the origin, which does not sit against $\mu(K)\ge\mu_{\min}$; on the compact branch the same scale is
the endpoint of the support, at which the density diverges, so that the extremal area of the
configuration is where the measure is carried. Proposition~\ref{prop:branches} is the passage between
the two, and it is the analytic substitution $\alpha^{2}\mapsto-\alpha^{2}$ rather than an identity of
measures.
The identification of Theorem~\ref{thm:ansoldi} is therefore exact on the flat sector, which is
where the degenerate configurations of Sections~\ref{sec:tachyon} and~\ref{sec:degenerate} live,
and where the areal reduction of \cite{Ansoldi} is itself valid. Off it the loop space
theory is the coarser of the two, seeing only the shadow of a configuration which the present
construction resolves. That the finer theory should distinguish configurations the quenched one
collapses is the expected relationship between a theory and its areal reduction. The cleanest witness
of it is a pair of coaxial circles, for which every interpolating surface has $b=0$ while the catenoid
spanning them has area growing with their separation.
\end{remark}

\section{The Two String Correlator}\label{sec:tachyon}

In \cite{Ansoldi}, the authors point out that their approach, fundamentally dealing with extended
relativistic objects, employs an equivalent classical action to the Polyakov CFT \cite{Polchinski},
but has not been demonstrated to have the same quantum spectrum. Whether the two theories have the
same spectrum is not settled here. No such equality is assumed.
The Ansoldi object is a contour-space Green kernel, while the expression below is the conventional
on-shell two-point amplitude. The calculation in this section compares their area transforms and
shows a special agreement of the degenerate tail at $D=4$. The two point scattering amplitude
\cite{Erbin} for the tachyon is written in Equation~\eqref{eq:tachyon}.
\begin{equation}\label{eq:tachyon}
\mathcal{A}_{2}=2\sqrt{sT+\vk^{2}}\,(2\pi)^{D-1}\delta^{D-1}(\vk'-\vk)
\end{equation}
Note that in Equation~\eqref{eq:tachyon} the brane tension $T$ (such that $m^{2}=sT$) is positive
while the multiplicative factor $s$ is negative. We instead compute its area transform directly
and compare the resulting tail with the Ansoldi/fixed-side law.

\begin{theorem}\label{thm:tachyon}
Let $D=4$, and let $\sigma_{i}$ and $\sigma_{j}$ be two `small' simplicial $2$ complexes,
concentrated about $\epsilon$ balls containing two points. In the limit $\epsilon\to0$, hence
$b\to0$, the modulus of the inverse tension transform of Equation~\eqref{eq:tachyon} has the same
$A^{-3}$ dependence as the degenerate tail of the Ansoldi/Fisk area law, up to normalization and the
unimodular phase of Equation~\eqref{eq:phase}. In general dimension the two area tails need not agree:
the transformed conventional two-point amplitude scales as
$|A|^{-(D+2)/2}$, whereas the Ansoldi/fixed-side family scales as $A^{-(D-1)}$.
\end{theorem}
\begin{proof}
Begin with Equation~\eqref{eq:tachyon}. One may take a Fourier transform $T\to A$ to obtain
Equation~\eqref{eq:tach1}. Writing $\sqrt{sT+\vk^{2}}=\sqrt{s}\sqrt{T+\vk^{2}/s}$ and translating in
$T$, it is the Fourier transform of $\sqrt{T}\mathbf{1}_{T>0}$ against a linear phase, and
$\int_{0}^{\infty}T^{\frac12}e^{-iAT}dT=\Gamma(\frac32)e^{-3i\pi/4}|A|^{-\frac32}$ as an oscillatory
integral, the branch being fixed by $\mathrm{Re}(iA)>0$.
\begin{equation}\label{eq:tach1}
\mathcal{A}_{2}(\vk,\vk')=\Four|^{T}_{A}\left(-2(1+i)\frac{e^{-i\frac{k^{2}A}{s}}\sqrt{-s}
(\mathrm{sgn}(A)+1)}{4|A|^{\frac32}}(2\pi)^{D-1}\delta^{D-1}(\vk'-\vk)\right)
\end{equation}
Both external momenta are then transformed,
$\vk\to\vx$ and $\vk'\to\vx'$; the delta of Equation~\eqref{eq:tach1} removes one of the two
integrations and leaves a single Gaussian in the common momentum. If the number of dimensions of
$\vk$ is $3$ (which it is in the context of our observed dimensionality), then
Equation~\eqref{eq:tach2} is obtained. The count is valid in every dimension: the transform in $T$
contributes $|A|^{-\frac32}$ and the transform in the $D-1$ momenta contributes
$|A|^{-\frac{D-1}{2}}$, for a total exponent $\frac{D+2}{2}$, of which the displayed
Equation~\eqref{eq:tach2} is the case $D=4$.
\begin{equation}\label{eq:tach2}
\mathcal{A}_{2}(\vx,\vx')=C\,\Four|^{T}_{A}\left(-2(1+i)
\frac{e^{-i\frac{\pi^{2} |\vx-\vx'|^{2}s}{A}}s^{2}(\mathrm{sgn}(A)+1)}{4|A|^{3}}\right)
\end{equation}
An extra factor of $|A|^{\frac32}$ emerges on the bottom due to the multivariate Fourier transform,
which is sensitive to the dimension. Explicitly,
$\int_{\R^{3}}e^{-iA\vk^{2}/s}e^{2\pi i\vk*(\vx-\vx')}d^{3}\vk
=(\pi s/iA)^{\frac32}e^{i\pi^{2}s|\vx-\vx'|^{2}/A}$,
which supplies both the extra $|A|^{-\frac32}$ and the phase
\begin{equation}\label{eq:phase}
\Theta(A,\vx-\vx')=\exp\Big(i\pi^{2}\frac{s|\vx-\vx'|^{2}}{A}\Big),\qquad|\Theta|=1.
\end{equation}
Here $C$ is some constant independent of $A$ or $s$ that occurs due to the Fourier transform. Now
consider restating the equation for the $2$ simplicial complex propagator in
Equation~\eqref{eq:simpprop}.
\begin{equation}\label{eq:simpprop}
\Four|^{T}_{A}\left(\chi^{\frac32}_{2}\Big(\frac{A}{\sqrt{b}}\Big)\right)
=\Four|^{T}_{A}\left(\frac{A(\mathrm{sgn}(A)+1)}{(b+2A^{2})^{2}}\right)
\end{equation}
Now, assume the conditions in the hypothesis of Theorem~\ref{thm:tachyon}, namely, that $\sigma_{i}$
and $\sigma_{j}$ are contained in $\epsilon$ balls which can be made arbitrarily small. Then $b\to0$
with $\epsilon\to0$. In this setting, the argument of the Fourier transform in
Equation~\eqref{eq:simpprop} is equivalent to that in Equation~\eqref{eq:tach2} barring normalization
and phase. Indeed, uniformly on compact subsets of $(0,\infty)$,
\[
\frac{A}{(b+2A^{2})^{2}}\xrightarrow[b\to0]{}\frac{1}{4A^{3}},
\]
which is the modulus of the argument of Equation~\eqref{eq:tach2} up to the constant $C$ and the
unimodular factor $\Theta(A,\vx-\vx')$ of Equation~\eqref{eq:phase}. Note that $b\to0$ is exactly the
degeneration in which the areal invariant
$b=(\Sigma^{\mu\nu}(\sigma_{j})-\Sigma^{\mu\nu}(\sigma_{i}))(\Sigma_{\mu\nu}(\sigma_{j})-\Sigma_{\mu\nu}(\sigma_{i}))$
of the two boundary loops vanishes, so that the scale parameter $\alpha=\sqrt{b/2}$ of the Fisk law of
Theorem~\ref{thm:ansoldi} vanishes with it and only the tail $A^{-3}$ survives; the two point tachyon
amplitude sees the Fisk law only through that tail. In this degeneration $\mu_{\min}$ and
$\sqrt{b/2}$ vanish together: both contours lie in $\epsilon$ balls, so $\mu_{\min}\to0$ directly,
and $\sqrt{b/2}\le\mu_{\min}$ by Proposition~\ref{prop:signedbound}.
\end{proof}

\subsection{Why the conventional two-point tail is not used to tune the geometry}\label{sec:tailconstraint}

The calculation above instead identifies two different dimensional scalings. The Ansoldi/fixed-side
area density is
\[
f_{\mathrm A}(A)\sim A^{-(D-1)},
\]
with survival index $D-2$, while the conventional on-shell two-point amplitude gives
\[
f_{\mathrm{tach}}(A)\sim A^{-(D+2)/2},
\]
with survival index $D/2$. These powers coincide at $D=4$ and only there. This is not treated as a
failure of the fixed-side geometry: the two calculations concern different observables, and the
Ansoldi contour-space kernel is the object that the geometric construction matches exactly in general
dimension.

There is nevertheless a suggestive dynamical connection. Section~\ref{sec:numerics} directly observes
that when the heavy-tailed excursion is coherent across neighboring spines, as expected when large
tension concentrates the measure near a minimal surface, the area inherits the arm survival index
$d-2$. When the heavy-tailed spine directions are decorrelated so that the face area becomes
quadratic in the excursion, the measured survival index is halved to $(d-2)/2$. We conjecture that
lowering the tension may weaken minimal-surface concentration and permit a crossover toward this
decorrelated regime. At $d=D$ its density tail is $A^{-D/2}$: it has the same one-half-per-dimension
slope as the conventional tachyon tail, but remains one power of $A$ away from the exact
$A^{-(D+2)/2}$ law. The remaining power could arise from the projection from a contour-space kernel
to an on-shell string state, or from another operation not represented in the present surface
surrogate. This last step is conjectural. The numerical experiment tests decorrelation directly; it
does not vary $T$ and therefore does not by itself establish the proposed lower-tension crossover.

\section{The JT Gravity Density of States}\label{sec:jt}

The third of the three expressions of Equation~\eqref{eq:three} is Equation 3.16 from \cite{SYY}.
Stanford, Yang, and Yao found expressions for the leading terms of Weingarten sums derived from
calculations in JT gravity. A complete description of this model of gravity is beyond the purview of
this paper; the point is that its correlation functions are modelled well by expectations of random
unitaries with Haar measures, which are related to Weingarten functions. Stanford et al.\ found a
power spectrum $\rho(s)$ which is equal to the Fisk density under integral transform. Let us begin by
writing the expectation value of the unitary function in JT gravity in Equation~\eqref{eq:jt}.
\begin{equation}\label{eq:jt}
\mathbb{E}[1]_{\mathrm{disk}}=e^{S_{0}}\int_{0}^{\infty}ds\,\frac{s}{2\pi^{2}}\sinh(2\pi s)\,
e^{-\beta\frac{s^{2}}{2}}
\end{equation}

\begin{theorem}\label{thm:jt}
Let $\rho(s)=\frac{s}{2\pi^{2}}\sinh(2\pi s)$ and let
$Z(\beta)=\int_{0}^{\infty}\rho(s)e^{-\beta s^{2}/2}ds$. Then for every $\beta>0$
\begin{equation}\label{eq:Zbeta}
Z(\beta)=\frac{1}{\sqrt{2\pi}}\beta^{-\frac32}e^{\frac{2\pi^{2}}{\beta}}.
\end{equation}
Consequently, with $m^{2}=1$ and $b=8\pi^{2}$, the substitution $W=i\beta$ carries the Schwinger
integrand of Equation~\eqref{eq:ansoldi} to $(-1)^{\frac32}(2\pi)^{-1}e^{\frac{\beta}{2}}Z(\beta)$,
and therefore, by Theorem~\ref{thm:ansoldi},
\begin{equation}\label{eq:jtfisk}
\int_{0}^{\infty}dW\,e^{-\frac{iW}{2}}\Big(\frac{1}{2i\pi W}\Big)^{\frac32}e^{\frac{i8\pi^{2}}{4W}}
=\frac{-1}{8\pi^{4}}\Four|^{1}_{A}\big(f^{\mathrm{odd}}_{2\pi}\big),\qquad
f_{2\pi}(A)=\frac{8\pi^{2}A}{(4\pi^{2}+A^{2})^{2}}\mathbf{1}_{A>0}.
\end{equation}
The area measure attached to the JT disk is thus the Fisk density of scale $2\pi$, the
identification being between the Ansoldi integrand and $e^{\beta/2}Z(\beta)$, that is, between the
Ansoldi integrand and the JT disk partition function computed with the ground state energy shifted by
$-\tfrac12$.
\end{theorem}

\begin{remark}[Where the dimension comes from here]\label{rem:jtdim}
JT gravity is two dimensional, so the exponent $\frac32$ of Equation~\eqref{eq:Zbeta} cannot be
the count $\frac{D-1}{2}$ of momenta transverse to a loop in $D=4$. It is instead the three zero
modes removed by the quotient $\mathrm{Diff}(S^{1})/SL(2,\R)$ of the Schwarzian boundary theory,
which contribute $\beta^{-\frac32}$. The specialisation to $D=4$ made by Theorem~\ref{thm:jt} is
therefore not read off the target space of the JT model, and it is not arbitrary either:
$SL(2,\R)\cong SO(2,1)$ is the isometry group of $AdS_{2}$, and $AdS_{2}$ arises as a factor of the
near horizon geometry of near extremal black holes in four dimensions, of which the Schwarzian is the
boundary mode. The three settings of Equation~\eqref{eq:three} therefore reach the exponent
$\frac32$ by three different routes, and the agreement of Section~\ref{sec:ansoldi} with
Section~\ref{sec:jt} is an agreement of the counts and not of the target dimensions.
\end{remark}

The identification of Theorem~\ref{thm:jt} is at the level of the disk alone, and one further
property should be recorded as a requirement rather than a result. A measure over surfaces which
reproduces a gravitational partition function ought to reproduce its stitching as well: cutting a
surface along an internal curve and summing over the boundary data on that curve should return the
uncut measure, and the propagators of Section~\ref{sec:ebp} should compose accordingly. The measure
of Section~\ref{sec:dictionary} is defined by a product over the spines of a skeleton, so a cut which
severs the skeleton transversally does factor; what is not established is that the sum over the cut
data reproduces the correct weight, nor that the higher genus and multi boundary amplitudes of
\cite{SYY} are obtained by the same prescription. It is conjectured that they are, and that the
single string dynamics identified here is one component of a compositional structure rather than the
whole of the correspondence.

\begin{proof}
For Equation~\eqref{eq:Zbeta}, note that $s\mapsto s\sinh(as)$ is even, so
\[
\int_{0}^{\infty}s\sinh(as)e^{-\frac{\beta s^{2}}{2}}ds
=\frac14\int_{-\infty}^{\infty}s\big(e^{as}-e^{-as}\big)e^{-\frac{\beta s^{2}}{2}}ds
=\frac12\int_{-\infty}^{\infty}se^{as-\frac{\beta s^{2}}{2}}ds
=\frac12\sqrt{\frac{2\pi}{\beta}}\frac{a}{\beta}e^{\frac{a^{2}}{2\beta}},
\]
the two exponentials contributing equally after $s\mapsto-s$. Putting $a=2\pi$ and multiplying by
$\frac{1}{2\pi^{2}}$ gives
$\frac{1}{2\pi^{2}}\cdot\pi\sqrt{2\pi}\beta^{-\frac32}e^{\frac{2\pi^{2}}{\beta}}
=(2\pi)^{-\frac12}\beta^{-\frac32}e^{\frac{2\pi^{2}}{\beta}}$.

For the second statement, substitute $W=i\beta$ in the integrand of Equation~\eqref{eq:ansoldi} at
$m^{2}=1$. Then $(2i\pi W)^{-\frac32}=(-2\pi\beta)^{-\frac32}=(-1)^{-\frac32}(2\pi)^{-\frac32}\beta^{-\frac32}$
and $e^{ib/(4W)}=e^{b/(4\beta)}=e^{2\pi^{2}/\beta}$ once $b=8\pi^{2}$. Comparing with
Equation~\eqref{eq:Zbeta}, the ratio of the two is $(-1)^{-\frac32}(2\pi)^{-1}e^{\beta/2}$.
Equation~\eqref{eq:jtfisk} is then Theorem~\ref{thm:ansoldi} at $m^{2}=1$, $b=8\pi^{2}$, for which the
Fisk scale is $\sqrt{b/2}=2\pi$ and $\pi^{2}b=8\pi^{4}$.
\end{proof}
\begin{remark}[On the factor $e^{\beta/2}$]\label{rem:shift}
The factor is not a constant and cannot be absorbed into a normalisation, so the identification is not
between the Ansoldi integrand and $Z(\beta)$ but between the Ansoldi integrand and
$e^{\beta/2}Z(\beta)$. It is nevertheless the identification one wants, for the following reason.
Writing the Schwarzian density of states in the energy variable $E=s^{2}/2$,
\[
e^{\frac{\beta}{2}}Z(\beta)=\int_{0}^{\infty}\rho(s)e^{-\beta\left(\frac{s^{2}}{2}-\frac12\right)}ds,
\]
so that the missing factor is exactly a shift of the ground state energy by $-\frac12$. In the Ansoldi
expression that shift is carried by the mass term $e^{-im^{2}W/2}$, and $m^{2}$ is precisely the
parameter which Section~\ref{sec:ansoldi} relates linearly to the brane tension. The additive constant
in the Hamiltonian is a convention in the Schwarzian theory and is fixed differently by different
authors; what Theorem~\ref{thm:jt} asserts is that with the convention $E\mapsto E-\frac12$ the two
expressions agree, the Fisk scale $2\pi$ being unaffected because the shift multiplies the transform
by a $\beta$-dependent exponential and not by a function of the area.
\end{remark}

\section{The Fisk Law as the Law of a Triangle's Area}\label{sec:triangle}

The Ansoldi loop-space calculation identifies the arbitrary-dimensional area family of
Equation~\eqref{eq:gammaD}, whose $D=4$ member is the Fisk law. This section identifies the same
family geometrically. The fixed-side triangle of
Proposition~\ref{prop:triangle} is the central construction and gives the Ansoldi law
exactly, dimension for dimension: the boundary datum fixes the short side, while the transverse
excursion of a single arm carries the two-dimensional reading of Lemma~\ref{lem:mshift}. In the language of Section~\ref{sec:reflection}, the sides are drawn
from the marginal spherical measure $\rho^{\frac{D}{2}-1}_{L}$ itself, whose density is
$(L^{2}-|\vx|^{2})^{\frac{D-3}{2}}_{+}$, and the $D-2$ derivatives which carry that measure to the
propagator are carried outside the area computation, exactly as the product
$\prod_{j}(ms_{p}+\va_{j}*\nabla)$ is carried outside the continuum multinomial in Theorem C.

The random vector of these ``arms'' follows the marginal spherical distribution
$(L^{2}-|\vx|^{2})^{\frac{D-3}{2}}$. Presuming, once again, that ``most'' triangulations remain close
to the plane spanned by these arms, we have that the random area drawn from this set is roughly the
random area drawn from a triangle whose arms are so distributed.

\begin{proposition}\label{prop:triangle}
Let $\vW$ be a random vector in $\R^{2}$ whose law is the two dimensional reading, in the sense of Lemma~\ref{lem:mshift}, of the arm system of
Equation~\eqref{eq:mvalue}, that is with density proportional to
$(1+\beta|\vec{w}|^{2})^{-\frac{d}{2}}$ for some $\beta>0$, equivalently the $q$-Gaussian at
$q_{2}=1+\frac{2}{d}$. Then:
\begin{enumerate}
\item $|\vW|^{2}$ has density $\propto(1+\beta y)^{-\frac{d}{2}}$ on $y>0$; equivalently
$\frac{d-2}{2}\beta|\vW|^{2}$ has the Fisher Snedecor $F(2,d-2)$ law, and at $d=4$ this is $F(2,2)$,
the ratio of two independent unit exponentials.
\item $|\vW|$ has density $\propto\rho(1+\beta\rho^{2})^{-\frac{d}{2}}$, which is
$\chi^{q_{2}}_{2}$ by Proposition~\ref{prop:index}, and which is the Fisk density $f_{\alpha}$ of
scale $\alpha=\beta^{-\frac12}$ if and only if $d=4$.
\item Consequently, if $\sigma_{j}$ is a segment of length $\ell$ and the arm joining $\sigma_{i}$ to
it has transverse displacement $\vW$ in the two plane orthogonal to that segment, then the area
$A=\frac{\ell}{2}|\vW|$ of the triangle they span has the same law with $\beta$ replaced by
$(2/\ell)^{2}\beta$, that is scale $\frac{\ell}{2\sqrt{\beta}}$ and exponent $\frac{d}{2}$; at $d=4$
this is the Fisk density of scale $\frac{\ell}{2\sqrt{\beta}}$.
\item With $d=D$ and $\ell/(2\sqrt{\beta})=\alpha$, this is exactly the Ansoldi family
$A(\alpha^{2}+A^{2})^{-D/2}$ of Corollary~\ref{cor:generalD}. Thus the fixed-side triangle matches
the Ansoldi area law.
\end{enumerate}
\end{proposition}
\begin{proof}
For (1), the law of $\vW$ is rotationally symmetric, so
$\mathbb{P}(|\vW|\in d\rho)\propto\rho(1+\beta\rho^{2})^{-\frac{d}{2}}d\rho$, and substituting
$y=\rho^{2}$ gives $\mathbb{P}(|\vW|^{2}\in dy)\propto(1+\beta y)^{-\frac{d}{2}}dy$. Thus
$\beta|\vW|^{2}$ has the $\beta'(1,\frac{d-2}{2})$ law after normalisation, so multiplication by $\frac{d-2}{2}$ gives $F(2,d-2)$; at $d=4$ the density
$(1+y)^{-2}$ is that of $F(2,2)$, equivalently of $U/V$ for $U,V$ independent unit exponentials. For
(2), $\mathbb{P}(|\vW|>\rho)=(1+\beta\rho^{2})^{-\frac{d-2}{2}}$, which is the Fisk survival function
$(1+(\rho/\alpha)^{2})^{-1}$ at $\alpha=\beta^{-1/2}$ exactly when $\frac{d-2}{2}=1$; differentiating
gives $f_{\alpha}$. Part (3) is (2) together with the fact that the family of
Equation~\eqref{eq:qchi} at $k=2$ is closed under positive scaling, as in Lemma~\ref{lem:det}.
\end{proof}

\begin{remark}[Dimension-for-dimension match with Ansoldi]\label{rem:crossing}
Proposition~\ref{prop:triangle} settles the fixed-side configuration in every dimension relevant to
the Ansoldi comparison. The one-dimensional invariant is $m=d-2$; its transverse two-dimensional
reading has exponent $(m+2)/2=d/2$; and Corollary~\ref{cor:generalD} gives the Ansoldi exponent
$D/2$. Hence $d=D$ matches the two laws identically. At $D=4$ the common member happens to be the
Fisk law, but four dimensions are not required for the geometric--Ansoldi identification.

\end{remark}

\begin{remark}[Decorrelated two-arm behavior is a diagnostic, not the matching law]\label{rem:decorrelated}
If neighboring heavy-tailed spine directions decorrelate so that both the spine length and the spine
separation scale with the same large excursion, the face area becomes quadratic in that excursion and
its survival index is halved. Section~\ref{sec:numerics} observes this regime directly. It is not used
to replace Proposition~\ref{prop:triangle}, and a generic two-random-arm Gram law is not an exact
Fisk/Ansoldi law. Its role here is diagnostic: loss of the coherence associated with minimal-surface
concentration changes the tail scaling, a fact used only in the conjectural comparison with the
conventional tachyon amplitude in Section~\ref{sec:tailconstraint}.
\end{remark}

\begin{remark}\label{rem:branchcare}
One point of care concerns the branch. The marginal spherical measure $\rho^{\frac{D}{2}-1}_{L}$ is
supported in the ball $|\vx|\le L$, and no compactly supported arm law can produce a Fisk area, since
the Fisk law has unbounded support; the identification of the transverse displacement with a
$q$-Gaussian at $q=\frac32$ is therefore an identification on the branch $q>1$ of
Equation~\eqref{eq:gausq}, that is, with the heavy tailed continuation of the marginal spherical family
rather than with the compactly supported member itself. This is not an accident of bookkeeping but is
the same reflection recorded in Proposition~\ref{prop:reflection}: it is the derivatives, and not the
measure, which carry the sign of the exponent. Proposition~\ref{prop:branches} makes the passage
between the two branches explicit: the compactly supported arm produces the area law
$g^{\gamma}_{\alpha}$ of Equation~\eqref{eq:compact}, which is carried to $f^{\gamma}_{\alpha}$ by
$\alpha^{2}\mapsto-\alpha^{2}$ at the same $\gamma$. The two routes therefore agree at the
level of the exponent, and the computation may be performed on whichever branch is convenient.
The substitution does not move $\gamma$, so it is not what carries the derivated marginal
spherical member at $\gamma=\frac32$ to the Fisk member at $\gamma=2$; that half unit is the reading
of Lemma~\ref{lem:mshift}, as recorded after Proposition~\ref{prop:branches}.

\end{remark}

The geometric setting in which Proposition~\ref{prop:triangle} applies is as follows. In a previous
portion of the proof, it was demonstrated that $K_{l^{*}_{2}}$ at high $T$, imaginary time $\tau$, and
fixed $\mu_{\min}$ is concentrated about those $(\sigma_{i},\sigma_{j})$ with
$\mu(K^{*})\sim\mu_{\min}(\sigma_{i},\sigma_{j})$. The
Taylor series for the continuum multinomial allows us to relax each of these assumptions; the Taylor
series expansion of $K_{l^{*}_{2}}$ in terms of the degrees of freedom of $(\sigma_{i},\sigma_{j})$
arranges itself into an expansion in terms of $\mu_{\min}$ which is valid when these
assumptions are relaxed.
This allows one to choose $\sigma_{i}$ and $\sigma_{j}$ more freely; allow $\sigma_{i}$ to be a single
vertex in $\R^{d}\times\R$ and $\sigma_{j}$ to be a rectangle in $\R^{d}$ with two sides of length
$\epsilon\ll1$. Let $\sigma_{i}$ be close enough in time to $\sigma_{j}$ that those
$K\in\Gamma_{\sigma_{i},\sigma_{j}}$ interpolating $\sigma_{i}$ and $\sigma_{j}$ are concentrated about
the convex cone projected from $\sigma_{i}$ to $\sigma_{j}$ (for this also consider leaving unrelaxed
the condition $T\gg1$). For $K\in\Gamma_{\sigma_{i},\sigma_{j}}$ the constraint that it `look like' a
cone and that it interpolates to $\sigma_{j}$ may be made rigorous in terms of the Hausdorff distances;
it implies that the distribution of its area is two times the distribution of the area of the
approximately triangular worldsheet spanned by two such arms.

The arms of that triangle, while concentrated around the arms of the triangular cone, have variable
lengths and displacements in space. Namely, the distribution of their transverse displacements is
determined by the marginal spherical distribution \cite{ODwyerQFT} with $\tau$ determined by the
endpoints in $\sigma_{j}$. By Equation~\eqref{eq:mvalue} that arm system has $m=d-2$, and by
Lemma~\ref{lem:mshift} its two dimensional reading, which is the one
Proposition~\ref{prop:triangle} requires, carries
exponent $\frac{d}{2}$. Since $\sigma_{j}$ has two sides of length $\epsilon\ll1$, the two arms of a
face differ by a vector of length of order $\epsilon$, and the area of that face is $\frac{\epsilon}{2}$
times the transverse displacement of a single arm; Proposition~\ref{prop:triangle} then gives its pdf as
$\chi^{\frac32}_{2}$, subject to the identification of the branch discussed in
Remark~\ref{rem:branchcare}. This is the fixed-side reading of
Remark~\ref{rem:crossing}; with $d=D$ it carries exponent $D/2$ and therefore agrees with the
Ansoldi family in every dimension. For this choice of $\sigma_{i}$ and $\sigma_{j}$ the distribution
is the desired distribution of areas, and therefore by the invariance argument of
Section~\ref{sec:reduction} it is the distribution for any $\sigma_{i},\sigma_{j}$.

\begin{remark}[Where the linearity comes from]\label{rem:linearity}
It is worth isolating the feature of this configuration on which the conclusion turns, since
Section~\ref{sec:numerics} measures it. The face area is the product of the length of a spine with the
separation of the two spines bounding it. Here the separation is fixed by the boundary datum
$\sigma_{j}$ at order $\epsilon$ and does \emph{not} scale with the transverse excursion, so the face
area is \emph{linear} in that excursion and Proposition~\ref{prop:triangle} transfers the arm law to the
area law with its exponent intact. If instead the excursions of adjacent spines were to differ by an
amount proportional to the excursion itself, both factors would scale together, the face area would be
quadratic in the excursion, and the tail exponent of the area would be half that of the arm. The
numerical experiment of Section~\ref{sec:numerics} exhibits both regimes. The linear regime is the fixed-side configuration of
Proposition~\ref{prop:triangle} and is the regime used in the proof; by Remark~\ref{rem:crossing} it
matches Ansoldi dimension for dimension. The quadratic regime is retained only as a diagnostic of
spine decorrelation. It halves the tail index, as measured in Section~\ref{sec:numerics}, and motivates
the lower-tension conjecture in Section~\ref{sec:tailconstraint}; it is not an alternative exact
Ansoldi law.

\end{remark}

\section{From Two Boundary Curves to the Triangle}\label{sec:reduction}

The evidence presented in this section will be presented in the guise of a proof, though with
significant liberties taken at some steps. The liberties are marked as they occur. The object of the
section is the single statement that the area law does not depend on the boundary curves except through the minimal area $\mu_{\min}$; granted that statement,
Section~\ref{sec:triangle} computes the law once, on the most degenerate pair of boundary data
available, and it holds for all of them. Two results about
minimal surfaces are used. The first is quoted rather than proved.

\begin{theorem}[Obtained from \cite{DPM}, under the stability hypotheses stated there]\label{thm:dpm}
Let $M_{1}$ and $M_{2}$ be two manifolds embedded in $\R^{n+1}$ with the same boundary, and
suppose $M_{2}$ is the unique area minimiser spanning it. Then
\[
\kappa\min\{\mu_{n+1}(E)^{2},\ \mu_{n+1}(E)^{\frac{n}{n+1}}\}\le|\mu_{n}(M_{1})-\mu_{n}(M_{2})|
\]

where $E$ is the region enclosed between the two surfaces, so that its boundary is the symmetric
difference $M_{1}\Delta M_{2}=(M_{1}\setminus M_{2})\sqcup(M_{2}\setminus M_{1})$, and $\mu_{k}$ is
the Lebesgue measure for dimension $k$ sets.
\end{theorem}

Note that $E$ is the region between the two surfaces and not a piece of either of them, so that what
Theorem~\ref{thm:dpm} controls is a volume and not a separation. The two are not the same thing. A
long thin spike attached to $M_{2}$ costs almost no area but sits a definite distance away, so some
hypothesis is needed which forbids the competitor from becoming arbitrarily thin. With such a
hypothesis the area deficit does control the Hausdorff distance, and with an explicit exponent, which
is Theorem~\ref{thm:hausdorff}.

\begin{theorem}\label{thm:hausdorff}
Let $M_{1}$ be an $n$ dimensional manifold embedded in $\R^{n+1}$ and let $M$ be
the solution to Plateau's problem for $\partial M_{1}$, as in Theorem~\ref{thm:dpm}. Suppose in
addition that $M_{1}$ is nowhere thinner than a fixed scale, in the sense that for some
$\theta,r_{0}>0$
\begin{equation}\label{eq:density}
\mu_{n}\big(M_{1}\cap B(x,r)\big)\ \ge\ \theta\,r^{n}\qquad\text{for every }x\in M_{1}
\text{ and every }r\le r_{0}.
\end{equation}
Then
\[
d_{H}(M_{1},M)\le C|\mu_{n}(M_{1})-\mu_{n}(M)|^{a}
\]
for constants $C$ and $a>0$ depending only on the dimension and on $\theta$ and $\kappa$; one may
take $a=\frac{1}{2(n+1)}$, which for surfaces in $\R^{3}$ is $a=\frac16$.
\end{theorem}

\begin{proof}
Let $\delta=d_{H}(M_{1},M)$ and let $x\in M_{1}$ be a point at distance $\delta$ from $M$. The ball
$B$ of radius $\delta/2$ about $x$ then meets $M_{1}$ and misses $M$ entirely, so the whole of the
boundary of $E$ inside $B$ is a piece of $M_{1}$, of area at least $\theta(\delta/2)^{n}$ by
Equation~\eqref{eq:density}. Since $M$ does not enter $B$, that piece of $M_{1}$ is separated from
$M$ by a distance at least $\delta/2$ at every one of its points, so the region $E$ between the two
contains the union of the normal segments it issues, and $\mu_{n+1}(E)\ge c\,\delta^{n+1}$ for a
constant $c$ depending only on the dimension and on $\theta$. Note that the area bound alone does not
give this: Equation~\eqref{eq:density} does not forbid $E$ from being thin normal to $M_{1}$, and it
is the separation of $M_{1}\cap B$ from $M$, available because $B$ was chosen to miss $M$, which
supplies the thickness. The normal segments are disjoint for $\delta$ below the reach of $M_{1}$,
which is the one point at which the argument assumes more than Equation~\eqref{eq:density}. On the
other hand, when the area deficit is small the first of the two branches in
Theorem~\ref{thm:dpm} is the smaller one, so that theorem reads
$\mu_{n+1}(E)\le\kappa^{-1/2}|\mu_{n}(M_{1})-\mu_{n}(M)|^{1/2}$. Combining the two displays gives
$\delta^{n+1}\le C|\mu_{n}(M_{1})-\mu_{n}(M)|^{1/2}$, which is the claim.
\end{proof}

Equation~\eqref{eq:density} is the hypothesis which forbids the thin spike, and it holds in every case
this paper needs. It holds for any surface carrying a bound on its Sobolev norm, since such a bound
forces the surface to be a mild graph at a definite scale. It also holds for the polyhedral surfaces
of $\Gamma_{\sigma_{i},\sigma_{j}}$, which have no curvature to bound: provided the triangles of the
complex do not become arbitrarily thin as the complex is refined, a small ball about any point of the
complex either lies inside a single flat face, where the area is that of a plane, or meets several
faces, where it is larger still. This is the form in which Equation~\eqref{eq:density} is assumed
below.

\subsection{Concentration on the Minimal Surface}\label{sec:concentration}
For imaginary times, Equation~\eqref{eq:conc} gives the propagator $K_{l^{*}_{2}}(\sigma^{i},\sigma^{j})$. Here $\Lap|^{T}_{A}$ denotes a Laplace transform.
\begin{equation}\label{eq:conc}
\begin{split}
K_{l^{*}_{2}}(\sigma^{i},\sigma^{j})&=\Lap|^{T}_{A}\big(\{K\mid\mu(K)=A,\ K\in\Gamma_{\sigma_{i},\sigma_{j}}\}\big)\\
&\sim\Lap|^{T}_{A}\big(\{K\mid\mu(K)=A,\ A\in[\mu(K^{*}),\mu(K^{*})+\epsilon_{1}],\\
&\hphantom{{}\sim\Lap|^{T}_{A}\big(\{K\mid{}}\ K\in\Gamma_{\sigma_{i},\sigma_{j}}\}\big)+\mathcal{O}(e^{-T\epsilon_{1}})
\end{split}
\end{equation}
In the first line of Equation~\eqref{eq:conc}, $\{\}^{A}_{B}$ denotes the measure of world surfaces on
the lattice with area $A$. In the second line of Equation~\eqref{eq:conc}, $K^{*}$ denotes the minimal
Euclidean area surface interpolating $\sigma_{i},\sigma_{j}$, and the interval is an interval of
areas, so it is $\mu(K^{*})$ and not $K^{*}$ which appears as its endpoint. This surface is guaranteed
by isoperimetric results; it is the solution to Plateau's problem. This problem is, in and of itself,
extremely intricate and one comes across many issues with singularities in solving it.

If string tension $T$ is fixed to be $\gg1$, then the error in the second line of
Equation~\eqref{eq:conc} is $e^{-T\epsilon_{1}}$. By Theorem~\ref{thm:hausdorff}, this obtains
Equation~\eqref{eq:haus} where $C$ and $a$ are only functions of the dimension and $a>0$.
\begin{multline}\label{eq:haus}
\{K\mid\mu(K)=A,\ A\in[\mu(K^{*}),\mu(K^{*})+\epsilon_{1}],\ K\in\Gamma_{\sigma_{i},\sigma_{j}}\}\\
\subset\{K\mid d_{H}(K^{*},K)<C\epsilon^{a}_{1},\ K\in\Gamma_{\sigma_{i},\sigma_{j}}\}=\Gamma_{\mathrm{Haus}}
\end{multline}
One may choose $\epsilon_{1}$ small enough that all positive measure $K$ approximate $K^{*}$ in a
Hausdorff sense, where $d_{H}$ denotes the Hausdorff distance used in Theorem~\ref{thm:hausdorff}.

\subsection{The Signed Areal Invariant and the Minimal Area}\label{sec:monotone}
Write $\mu_{\min}(\sigma_{i},\sigma_{j})=\inf\{\mu(K)\mid K\in\Gamma_{\sigma_{i},\sigma_{j}}\}$
for the area of the Plateau solution. In addition, note that for any oriented manifold $K\in\Gamma_{\sigma_{i},\sigma_{j}}$ we have Equation~\eqref{eq:stokes}.
\begin{equation}\label{eq:stokes}
b^{\mu\nu}_{\sigma_{i},\sigma_{j}}=\int_{\sigma_{j}}x^{\mu}dx^{\nu}-\int_{\sigma_{i}}x^{\mu}dx^{\nu}
=\int_{K}d(x^{\mu}dx^{\nu})=\int_{K}dx^{\mu}\wedge dx^{\nu},
\qquad b_{\sigma_{i},\sigma_{j}}=b^{\mu\nu}_{\sigma_{i},\sigma_{j}}b_{\sigma_{i},\sigma_{j},\mu\nu}
\end{equation}
If the tangent bundle to $K$ is spanned by $\hat{e}_{\mu}$ and $\hat{e}_{\nu}$ everywhere on $K$,
that is, if $K$ is flat and consistently oriented, then the right hand side of
Equation~\eqref{eq:stokes} is the surface area of $K$. In general, however,
$\int_{K}dx^{\mu}\wedge dx^{\nu}$ is the area of the shadow of $K$ on
the $\mu\nu$ plane counted with orientation, so that a piece of $K$ whose normal has reversed
contributes negatively and cancels against the sheet it folds over. Equation~\eqref{eq:stokes} is
exact: it is Stokes' theorem, and it is the reason $b$ is a function of the boundary data alone.

\begin{proposition}[The signed area bound]\label{prop:signedbound}
Let $\sigma_{i},\sigma_{j}$ be disjoint contours and let $b=b_{\sigma_{i},\sigma_{j}}$ be the
invariant of Equation~\eqref{eq:stokes}. Then for every oriented
$K\in\Gamma_{\sigma_{i},\sigma_{j}}$,
\[
\sqrt{b/2}\ \le\ \mu(K),\qquad\text{hence}\qquad \sqrt{b/2}\ \le\ \mu_{\min}(\sigma_{i},\sigma_{j}),
\]
with equality if and only if $K$ is flat and consistently oriented. In $\R^{3}$ the statement
reads $|\vec{S}|\le\mu(K)$ for the vector area $\vec{S}=\int_{K}\hat{n}\,dA$, and
$\sqrt{b/2}=|\vec{S}|$.
\end{proposition}
\begin{proof}
Let $\tau$ denote the unit tangent $2$ vector of $K$, so that
$\int_{K}dx^{\mu}\wedge dx^{\nu}=\int_{K}\tau^{\mu\nu}d\mathcal{H}^{2}$ and $|\tau|=1$ pointwise,
where $|\cdot|$ is the Euclidean norm on the antisymmetric arrays, $|\xi|^{2}=\frac12\xi^{\mu\nu}\xi_{\mu\nu}$.
Then $\sqrt{b/2}=\big|\int_{K}\tau\,d\mathcal{H}^{2}\big|\le\int_{K}|\tau|\,d\mathcal{H}^{2}
=\mathcal{H}^{2}(K)=\mu(K)$, the middle step being the triangle inequality. Equality in the triangle
inequality forces $\tau$ to be constant $\mathcal{H}^{2}$ almost everywhere, which is flatness with a
consistent orientation; conversely a flat consistently oriented $K$ has
$\big|\int\tau\big|=\mu(K)|\tau|=\mu(K)$. Taking the infimum over $K$ gives the second display. In
$\R^{3}$ the three independent entries $b^{23},b^{31},b^{12}$ are the components of
$\vec{S}=\int_{K}\hat{n}\,dA$ and $b=b^{\mu\nu}b_{\mu\nu}=2|\vec{S}|^{2}$.
\end{proof}

The remainder of the section requires only that the measure at high tension depend on the
boundary data through the single number $\mu_{\min}(\sigma_{i},\sigma_{j})$, and that is exactly what
Section~\ref{sec:concentration} delivers: the concentration argument places the measure on surfaces
Hausdorff close to $K^{*}$, whose area is $\mu_{\min}$ by definition. Two pairs of boundary curves with the same minimal area are then carried
into one another by Section~\ref{sec:transport}, which is written in terms of minimal areas throughout,
and Section~\ref{sec:degenerate} exhibits a pair realising any prescribed minimal area. The invariant
$b$ retains two roles: it is the sharp lower bound $\sqrt{b/2}\le\mu_{\min}$, and it equals
$\mu_{\min}$ on the flat sector, which
is where Theorem~\ref{thm:ansoldi} and Remark~\ref{rem:quench} place the loop space propagator of
\cite{Ansoldi}.

At high tension $T$, the simplicial complex measure will concentrate about surfaces with least area. This means that
Equation~\eqref{eq:conc2} holds for sufficiently high $T$.
\begin{multline}\label{eq:conc2}
K_{l^{*}_{2}}(\sigma^{i},\sigma^{j})\sim\Lap|^{T}_{A}\Big(\{K\mid\mu(K)\in[\mu_{\min}(\sigma_{i},\sigma_{j}),
\mu_{\min}(\sigma_{i},\sigma_{j})+\epsilon_{1}],\ K\in\Gamma_{\sigma_{i},\sigma_{j}}\}\\
\times\mathbf{1}_{[\mu_{\min},\mu_{\min}+\epsilon_{1}]}(A)\Big)+\mathcal{O}(e^{-T\epsilon_{1}})
\end{multline}

\subsection{Volume Preserving Transport Between Boundary Data}\label{sec:transport}
Again Theorem~\ref{thm:hausdorff} shows that the set of $K$ in Equation~\eqref{eq:conc2} must `look like'
$K^{*}$ for those starting pairs $(\sigma^{i},\sigma^{j})$ that have a minimal surface area that also has
measure close to $\mu_{\min}(\sigma_{i},\sigma_{j})$. Assume that $(\sigma_{i},\sigma_{j})$ and
$(\sigma'_{i},\sigma'_{j})$ both have $K^{*}$ and $K^{*,\prime}$ which are within the range of
$\mu_{\min}(b)$ required in Equation~\eqref{eq:conc2}. Let $\mu(K^{*})\le\mu(K^{*,\prime})$. Let
$K_{\mathrm{intermediate}}$ be a surface which is $d_{H}$ close to $K^{*}$ and has Lebesgue measure equal
to $K^{*,\prime}$. Moser's theorem \cite{Moser} gives us a volume-preserving diffeomorphic map
$\phi:K_{\mathrm{intermediate}}\to K^{*,\prime}$. The $h$-principle \cite{CEM} then extends the domain of
$\phi$ to a small neighborhood about $K_{\mathrm{intermediate}}$. This implies that the full set of
worldsheets of Equation~\eqref{eq:haus} from $(\sigma_{i},\sigma_{j})$ map bijectively and volume
preservingly to manifolds interpolating $\sigma'_{i}$ and $\sigma'_{j}$. Since
$K_{l^{*}_{2}}(\sigma_{i},\sigma_{j})$ is a Fourier transform of the `number' of worldsheets of area $A$,
this is very close to demonstrating that it must be equal to $K_{l^{*}_{2}}(\sigma'_{i},\sigma'_{j})$.

The remaining problem is that $\phi$ is not piecewise linear, and will not map the original polyhedral
surfaces constituting the sum over $\Gamma_{\sigma_{i},\sigma_{j}}$ to other polyhedral surfaces in
$\Gamma_{\sigma'_{i},\sigma'_{j}}$. The mapping also must conserve $|K|$, i.e.\ it must retain the number
of linear segments in its skeleton. If, however, $\phi$ was $\epsilon$ close in $d_{H}$ metric to a linear
mapping about each face of $K$, then there is a unique polyhedron $K'\in\Gamma_{\sigma'_{i},\sigma'_{j}}$
to which the images of $K$ under $\phi$ map; this polyhedron (outside of some vanishing set of domain $K$)
would necessarily have the same number of segments in its $1$ simplex \emph{and} approximately the same
area.

If $\epsilon$ was arbitrary, then the result is gotten as there is an approximate bijection between
$\Gamma_{\sigma_{i},\sigma_{j}}$ and $\Gamma_{\sigma'_{i},\sigma'_{j}}$ which conserves $|K|$. When $|K|$
is large and $T\gg1$, each face in $K$ must arrange itself such that the span of its largest face is small
with respect to the scale by which $\phi$ is nonlinear. The well-defined limit over $m$ forces $|K|$ to
become large. Therefore $\epsilon$ is allowed to become small as $T\to\infty$ and the measures over
$\Gamma_{\sigma_{i},\sigma_{j}}$ and $\Gamma_{\sigma'_{i},\sigma'_{j}}$ are the same for equivalent
$\mu_{\min}$.

Here is the single significant liberty we take. Moser's theorem and the $h$-principle deliver the
volume-preserving diffeomorphism and its extension to a neighbourhood, and both are quoted correctly; what
is not established is that the induced correspondence on \emph{polyhedral} surfaces is a bijection which
preserves the segment count, rather than merely an approximate one whose error is controlled by a quantity
that also grows with the refinement. The author believes a more careful control of the error would
establish this argument completely. The step is in any case a mild one: $\phi$ is a diffeomorphism, hence
locally linear to first order at every point, and the faces in question are those of an arbitrarily fine
polyhedral surface, so what is being assumed is that a smooth map takes a sufficiently fine polyhedral
surface to a surface which is again well approximated by one of the same combinatorial size. What is
missing is a rate at which the two refinements may be taken together, not the plausibility of the
statement.

\subsection{The Degenerate Pair}\label{sec:degenerate}
Given the above, the reduction to Section~\ref{sec:triangle} is the following.
The boundary data of Section~\ref{sec:ebp} are a disjoint pair of one-manifolds $C_{1}\sqcup C_{2}$ with
$\partial M=C_{1}\sqcup C_{2}$, so the degenerate configuration must be presented within that class. Fix $a>0$, let $C_{2}=\partial\triangle$ be the boundary of a triangle $\triangle$ of area $a$
in a fixed two-plane of $\R^{d}$, and let $C_{1}^{\delta}$ be the boundary of the $\delta$-square
inscribed in one corner of $\triangle$ in that same plane. The pair
$(C_{1}^{\delta},C_{2})$ lies in $\Gamma$ for every $\delta>0$, and the planar annular region
$\triangle^{\delta}$ between the two loops has constant tangent $2$ vector, so
Proposition~\ref{prop:signedbound} holds for it with equality. Its area therefore attains the lower
bound $\sqrt{b/2}$ which that proposition places under every competitor, and it is area-minimising
without further argument; the Goldschmidt configuration which capped the two contours separately and
joined them by a thin neck costs $a+\delta^{2}$ against $a-\delta^{2}$ and does not compete. As
$\delta\downarrow0$ one has $\mu_{\min}(C_{1}^{\delta},C_{2})=\mu(\triangle^{\delta})=a-\delta^{2}\to
a$ with $\sqrt{b/2}=\mu_{\min}$ throughout, so this is the configuration on which the equality case of
Proposition~\ref{prop:signedbound} is exhibited. Scaling $\triangle$ sweeps $a$ over $(0,\infty)$, so
every value of $\mu_{\min}$ is realised within this family and Section~\ref{sec:transport} carries the
law from it to any other pair with the same minimal area.

This pair carries the picture of Section~\ref{sec:triangle} intact. As $\delta\downarrow0$ the
inner contour becomes a single vertex sitting in the chosen corner of $\triangle$, and the spines of
Section~\ref{sec:dictionary} emanate from it. Two of them lie closest to the two sides of $\triangle$
meeting at that corner, and it is these two which serve as the edges of the surface: together with the
opposite side of $\triangle$ they span a triangle, and a surface of area near $\mu_{\min}$ is to
lowest order exactly that triangle, its two legs being spines drawn from the one dimensional measure of
\cite{ODwyerQFT} and its third side fixed by the boundary datum. This is the configuration of
Section~\ref{sec:triangle}, the corner in the role of $\sigma_{i}$ and the opposite side in the role of
$\sigma_{j}$, and Proposition~\ref{prop:triangle} gives the resulting area law.
Let $\triangle_{\min+\epsilon}$ be another surface interpolating $C_{1}^{\delta}$ and
$C_{2}$, but
with $\mu(\triangle_{\min+\epsilon})\le\mu(\triangle^{\delta})+\epsilon$. Then it is a general rule of thumb
in the theory of extremal surfaces that $\exists\epsilon'>0$ such that if $\epsilon<\epsilon'$ then
$d_{H}(\triangle_{\min},\triangle_{\min+\epsilon})<C\epsilon$. For this pair of boundary sets the rule of
thumb is Theorem~\ref{thm:hausdorff} and needs no separate justification.

This means there is an open interval about its minimum wherein the random variable drawn from worldsheet
areas is equal to a random variable drawn from a simple triangle, with two random ``arms''. These arms are
themselves drawn in a manner equivalent to the piecewise segments of the one-dimensional construction of
\cite{ODwyerQFT}, and Proposition~\ref{prop:triangle} computes the resulting law of the area.

\section{Conclusion}\label{sec:conclusion}

In this paper, the one-parameter family $A(\alpha^{2}+A^{2})^{-D/2}$, of which the Fisk or
log--logistic density of shape $2$ is the member $D=4$, was identified as the area law of the Ansoldi
loop-space kernel and reproduced geometrically by the fixed-side triangle construction. The exponent is carried dimension for dimension on the Ansoldi/geometric comparison:
Equation~\eqref{eq:gammaD} gives $D/2$, while Equation~\eqref{eq:mvalue} and
Lemma~\ref{lem:mshift} give $d/2$, so $d=D$ with no additional correction. The conventional tachyon
two-point amplitude and the JT expression provide separate four-dimensional comparisons rather than
general-$D$ constraints on the geometry.

The loop space propagator of Ansoldi et al.\ was shown, in Theorem~\ref{thm:ansoldi}, to be exactly the
Fourier transform in the area variable of the Fisk density of scale $\sqrt{b/2}$ at $D=4$, where $b$
is the signed areal invariant of the two boundary loops; Corollary~\ref{cor:generalD} gives the
corresponding $A(\alpha^{2}+A^{2})^{-D/2}$ family in general dimension. The two-point tachyon amplitude
of \cite{Erbin} was shown, in Theorem~\ref{thm:tachyon}, to reproduce the same $A^{-3}$ degenerate tail
at $D=4$, up to normalization and an explicitly identified phase, while its general-$D$ tail is
$A^{-(D+2)/2}$ and is not identified with the Ansoldi kernel. The density of states of JT gravity found
by Stanford, Yang, and Yao was shown, in Theorem~\ref{thm:jt}, to integrate to the disk partition function
$\frac{1}{\sqrt{2\pi}}\beta^{-\frac32}e^{2\pi^{2}/\beta}$, which under $W=i\beta$ is the Schwinger integrand
of \cite{Ansoldi} with $m^{2}b=8\pi^{2}$ up to the factor $e^{\beta/2}$ identified in
Remark~\ref{rem:shift} as a shift of the ground state energy, so that its area law is the Fisk density of
scale $2\pi$.

Section~\ref{sec:triangle} then identified the same density geometrically, and
Proposition~\ref{prop:branches} showed that the geometric and the analytic identifications are the same
identification read on the two branches of a single family. The Fisk density and the compactly supported
$A(\alpha^{2}-A^{2})^{-2}$ are carried into one another by $\alpha^{2}\mapsto-\alpha^{2}$; the latter is the
area transform of the marginal spherical distribution of \cite{ODwyerQFT}, and the former is the law which
the Ansoldi propagator exhibits, with the tachyon and JT calculations meeting it in the
four-dimensional cases stated above. The area law of this paper is therefore not an analogy with the
one dimensional theory but its continuation.

Taken together this establishes, short of the steps isolated in Section~\ref{sec:reduction} and
Remark~\ref{rem:branchcare}, a generalisation of the one dimensional construction of \cite{ODwyerLattice} and
\cite{ODwyerQFT} to extended body quantum field theories. The one dimensional statement was that the
relativistic propagator is a transform, from the length of a rectifiable path to the mass as its conjugate
variable, of a genuine measure over directed paths, the derivated marginal spherical distribution supplying
the length law. The statement obtained here is the same one dimension up: the extended body propagator is a
transform, from the area of a worldsheet to the brane tension as its conjugate variable, of a measure over
surfaces interpolating two boundary contours, with the family
$A(\alpha^{2}+A^{2})^{-D/2}$ supplying the Ansoldi-matched area law and the Fisk density as its $D=4$
member. The minimal area $\mu_{\min}(\sigma_{i},\sigma_{j})$ occupies the position held by the endpoint
displacement in the one dimensional theory. The conventional tachyon
and JT calculations are retained as separate comparisons in the regimes stated above; in particular, the
tachyon calculation motivates the conjectural coherence crossover of Section~\ref{sec:tailconstraint}
rather than a modification of the exact Ansoldi/fixed-side match. The construction is generic in the sense
that it applies to point and extended body problems alike without modification of its underlying
combinatorial object, the continuum multinomial coefficient.

\section{Numerical Evidence for the Area Law}\label{sec:numerics}

This section reports a numerical experiment and is included so that the area law of
Sections~\ref{sec:ansoldi}--\ref{sec:triangle} may be checked against a direct simulation of the surface
sum rather than only against the transform identities. The whole of the code which produced it is
reproduced below, so that the experiment may be repeated exactly.

The geometry is deliberately not special. Two nontrivial end curves are fixed in $\R^{3}$, neither planar
nor concentric,
\[
C_{1}(t)=\big(\cos t,\ \sin t,\ \tfrac14\sin3t\big),\qquad
C_{2}(t)=\big(1.8\cos t,\ 1.2\sin t,\ 1+0.30\sin2t\big),
\]
for which the areal invariant of Equation~\eqref{eq:stokes} is $b=26.561$, so that the degeneration
$b\to0$ of Section~\ref{sec:tachyon} is nowhere near. $M$ spines join $C_{1}(t_{k})$ to $C_{2}(t_{k})$ at
$t_{k}=2\pi k/M$; each is subdivided into $N$ segments and carries a transverse excursion; adjacent spines
are paired step by step and every quadrilateral is split into two triangles. The reported area is the sum
over all $2MN$ triangles, so that what is measured is a granular surface sum and not the area of a single
triangle.

A word first on what is being simulated, since it is not the measure of
Section~\ref{sec:dictionary} itself. That measure assigns to a skeleton the product of $M$ marginal
spherical measures with the $d-2$ derivatives standing outside, and it resists direct
implementation for two reasons. The derivated object is a distribution and not a density, so it
carries no sampling procedure at all; and the undifferentiated member from which one would sample is
supported in a ball whose radius is the length of the spine, so that a faithful sampler must carry the
length of every spine as a random variable coupled to its neighbours through the boundary data. What
is simulated instead is a surrogate. The geometry is exactly that of Section~\ref{sec:dictionary}. The same skeleton, the same pairing of consecutive spines into faces, the same Nambu Goto sum of
their Gram determinants. Only the law carried by the skeleton is replaced, by the heavy tailed
member of the same one parameter family at the same exponent, which is the substitution of
Proposition~\ref{prop:branches} and the branch identification of Remark~\ref{rem:branchcare}.

The reason the surrogate ought to evaluate the measure it stands in for is the concentration statement
of Section~\ref{sec:concentration}. At high tension the measure of Section~\ref{sec:dictionary} is
carried by surfaces Hausdorff close to the minimal one, and on those surfaces the spine lengths are
pinned by the boundary data rather than free; the coupling which the surrogate drops is therefore the
coupling which the tension suppresses, and the sweep in $\sigma$ reported below is the numerical form
of that statement, the Fisk form emerging once the coherent
heavy-tailed mode dominates the fixed scale of the boundary geometry. The experiment should accordingly be read as evidence
for the heuristic of Section~\ref{sec:reduction} as much as for the area law itself. What it tests is
the one step of Section~\ref{sec:triangle} which the analytic sections do not supply: whether a
Nambu--Goto sum over a granular skeleton transmits an arm law of exponent $-2$ to an area law of
exponent $-2$, rather than distorting it; and it measures how far the skeleton may depart from the
concentrated regime before the transmission fails.

The one physical input is the law of the transverse excursion, and it is now a field on the skeleton
rather than a single vector. Write $W_{k}$ for the excursion of the $k$th spine. Then
\begin{equation}\label{eq:excursion}
W_{k}=\underbrace{\sigma\sqrt{\nu/S}\,Z}_{\text{collective}}
\;+\;\kappa\underbrace{\sum_{j=1}^{3}\frac{1}{j}\big(a_{j}\cos jt_{k}+b_{j}\sin jt_{k}\big)}_{\text{incoherent}},
\qquad S\sim\chi^{2}_{\nu},\quad Z,a_{j},b_{j}\sim N(0,I_{2}),
\end{equation}
with $S$ and $Z$ drawn once per surface and $a_{j},b_{j}$ drawn once per surface but independently of the
collective amplitude. Following Proposition~\ref{prop:triangle}, the collective amplitude
$\sigma\sqrt{\nu/S}\,Z$ is a two dimensional Student's $t$ with $\nu$ degrees of freedom, which at $\nu=2$
is exactly the $q$-Gaussian at $q=\frac32$ of Equation~\eqref{eq:gausq}. The excursion is a two
dimensional vector by construction, so the experiment is run in the planar reduction of
the planar reading throughout, and $\nu$ is the invariant $m$ of Lemma~\ref{lem:mshift}; the
value $\nu=2$ is $d=4$ by Equation~\eqref{eq:mvalue}, so the sweep in $\nu$ reported below is a sweep
in the ambient dimension of the arm system and the Fisk law is expected at $\nu=2$ alone. Three features of
Equation~\eqref{eq:excursion} are now stated explicitly.

First, the incoherent part is a \emph{smooth} field in the loop parameter, so adjacent spines differ by
$O(1/M)$ and increasing $M$ refines a single surface rather than injecting new roughness at every
refinement. This is what makes an $M$ sweep a convergence test.

Second, the incoherent amplitude $\kappa$ is independent of the collective amplitude, and this is the
feature on which the whole result turns. By Remark~\ref{rem:linearity}, the area of a face is the product
of the length of a spine with the separation of the two spines bounding it. If the separation is set by
the boundary geometry, or by an incoherent part which does not scale with the excursion, then the face
area is linear in the excursion and the arm law is transmitted to the area law with its exponent intact.
If instead the separation scales with the excursion, the face area is quadratic in it and the tail
exponent of the area is halved. Test C below measures exactly this.

Third, $\kappa=0$ gives a rigid transverse excursion. This is the regime which the degenerate configuration of Section~\ref{sec:triangle} realises, since there
the separation of adjacent spines is fixed by the boundary datum at order $\epsilon$. It is the regime in
which the Fisk law is expected, and Test D measures how far from it one may move.

\begin{figure}[h]
\centering
\includegraphics[width=\textwidth]{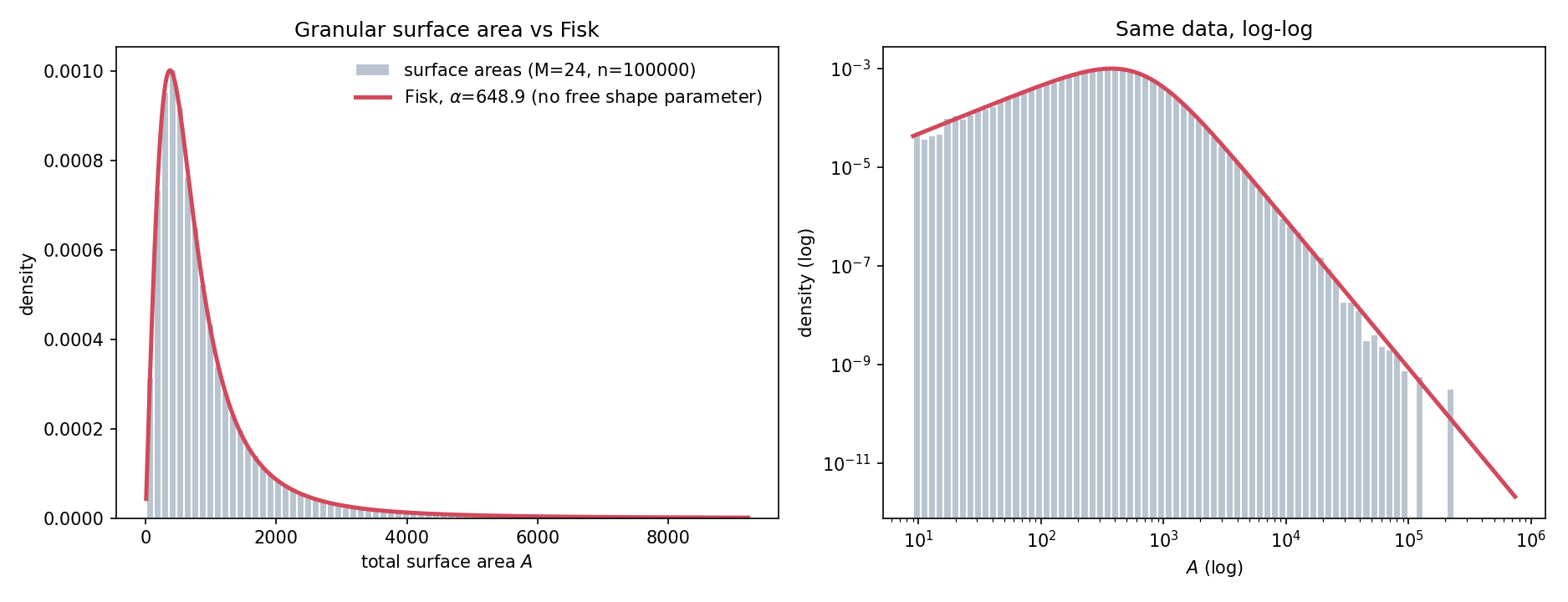}
\caption{Total area of the granular surface at $M=24$, $n=10^{5}$ samples, against the Fisk density of
Equation~\eqref{eq:fisk} with scale matched to the sample median and no free shape parameter. Left: linear
axes. Right: the same data on logarithmic axes.}
\end{figure}

Four quantities are reported, at $n=10^{5}$ samples and $N=10$. The sampling noise floor, that is, the
Kolmogorov--Smirnov distance of an exact Fisk sample of the same size, is $0.00273$.

The first is fidelity across $M$ at $\kappa=0$ and $\sigma=40$. At $M=6,12,24,48$ the surface areas give
$0.00278$, $0.00311$, $0.00253$, $0.00194$, i.e.\ ratios to the noise floor of $1.02$, $1.14$, $0.93$,
$0.71$: at every $M$ the sampled areas are not distinguishable from exact Fisk draws. The medians over the
same sequence are $565.0$, $641.5$, $648.9$, $652.0$; the area converges under refinement rather than
being independent of $M$ by construction, so the sequence is evidence about the limit of
Section~\ref{sec:dictionary} rather than about a single triangle sampled four times. The Hill estimates
of the tail index sit at $2.0$ throughout.

The second is the approach to that floor as the transverse fluctuation grows relative to the minimal
surface underneath it. At $M=24$ and $\kappa=0.25$ the ratio of the Kolmogorov--Smirnov distance to the
noise floor runs $68.50$, $13.95$, $1.46$, $0.95$, $1.45$, $1.99$, $1.47$ as $\sigma$ runs
$1,3,10,30,100,300,1000$, so the law reaches the floor by $\sigma\approx10$ and stays there over a further
two decades in scale.

The third is the one which makes the experiment informative rather than circular. The arm law is an input,
so it must be shown that the surface transmits it rather than manufacturing the answer, and it must be
shown that the transmission is a property of the geometry rather than of the sampler. Varying $\nu$ at
$M=24$ and $\kappa=0.25$, the tail index of the area is $0.98$, $1.52$, $1.93$, $2.70$, $3.62$ for
$\nu=1,\frac32,2,3,4$. The reported index is the mean of the Hill estimator over the three cuts
$k=500,2000,10000$ used in the script below, that is over tail fractions of one half, two and ten
percent; the deepest cut is biased low, and it accounts for the shortfall at $\nu=3$ and $\nu=4$. The
Kolmogorov--Smirnov distance from the Fisk law has a sharp minimum at $\nu=2$, namely $0.0905$,
$0.0350$, $0.00227$, $0.0356$, $0.0591$, a factor of fifteen to twenty-five below its neighbours. The Fisk
law is therefore not something the geometry produces for any input; it appears precisely when the arm law
is the $q=\frac32$ one.

The same sweep run with the collective mode removed, that is, with an independent heavy tailed amplitude
per spine, does not reproduce this. There the Kolmogorov--Smirnov distance from the Fisk law increases
monotonically through $\nu=2$ rather than having a minimum there, and in the cleanest version of the
control --- a scale shared across spines but an independent direction per spine, so that adjacent
separations scale with the excursion --- the measured tail index of the area is $0.49$, $0.76$, $0.99$,
$1.45$, $1.83$ for the same five values of $\nu$, which is $\nu/2$ to within the estimator's precision.
This is exactly the bilinear mechanism of Remark~\ref{rem:linearity}: with both the spine length and the
spine separation scaling with the excursion, the face area is quadratic in it and the exponent halves. This control is not used to replace the fixed-side law. It supplies the empirical basis for the
conjecture of Section~\ref{sec:tailconstraint}: if lower tension weakens concentration around the
minimal surface, heavy-tailed spine directions may decorrelate and drive the area from survival index
$d-2$ toward $(d-2)/2$. The simulation demonstrates the decorrelation effect itself, not its dependence
on tension. At $d=D$ this decorrelated density has tail $A^{-D/2}$, which has the same dimensional
slope as, but is one power shallower than, the conventional tachyon tail
$A^{-(D+2)/2}$.

The fourth quantifies that hypothesis. At $M=24$ and $\sigma=40$, the ratio of the Kolmogorov--Smirnov
distance to the noise floor runs $0.64$, $0.60$, $3.79$, $5.15$ as $\kappa$ runs $0,0.25,1,4$ against the
$O(1)$ scale of the boundary geometry. The Fisk law is thus stable under an incoherent deformation of the
skeleton up to roughly a quarter of the boundary scale and degrades smoothly beyond it, the tail index
remaining at $2$ throughout because the tail is carried by the collective mode alone.

\begin{figure}[h]
\centering
\includegraphics[width=\textwidth]{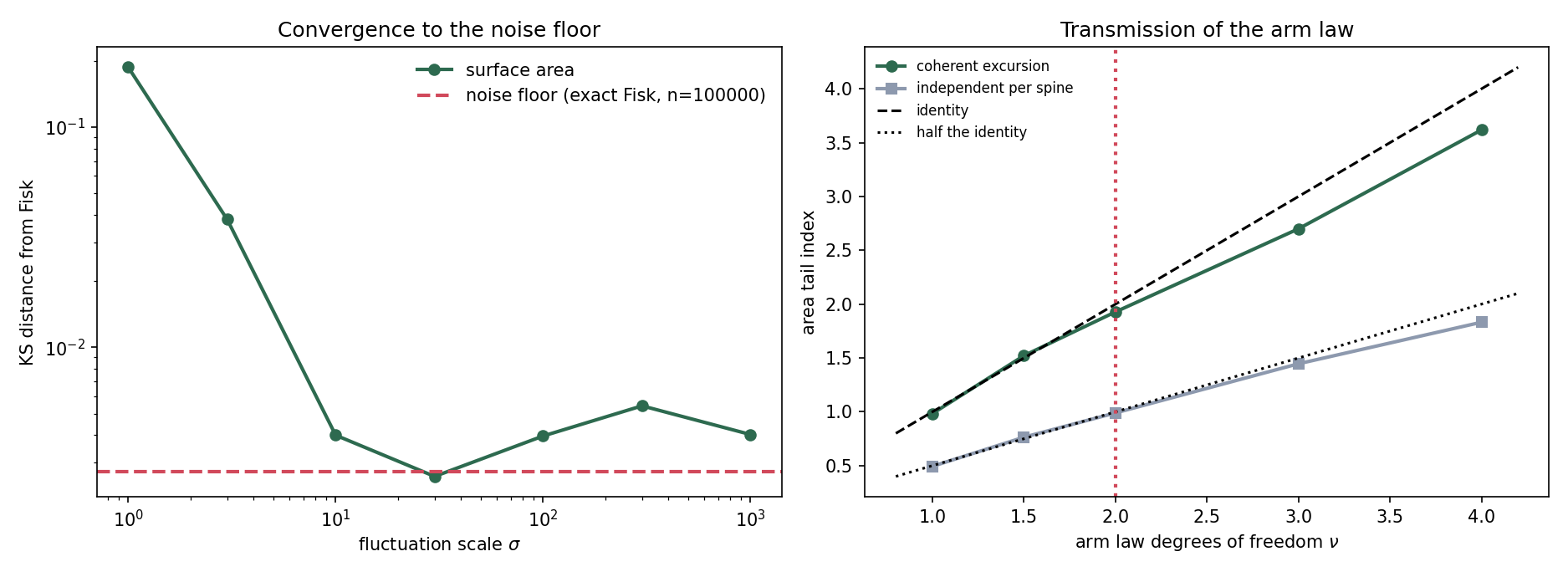}
\caption{Left: Kolmogorov--Smirnov distance from the Fisk law against the transverse fluctuation scale,
with the sampling noise floor marked. Right: measured tail index of the area against the degrees of freedom
$\nu$ of the arm law, for the coherent excursion and for the independent per-spine control, with the identity
line and half the identity line marked.}
\end{figure}

This establishes the surface half of the statement, in the following form. With the $q=\frac32$ arm law of
Proposition~\ref{prop:triangle}, and with the transverse excursion of the skeleton coherent up to a
deformation small against the boundary scale, the area of a granular surface spanned between two nontrivial
boundary curves, at nondegenerate $b$, with up to a thousand triangles, is the Fisk law to within sampling
error, and remains so as the triangulation is refined. The coherence hypothesis is not a modelling
convenience: it is the numerical content of the concentration statement of
Section~\ref{sec:concentration}, which places the measure at high tension on surfaces Hausdorff-close to the
minimal one, and it is the same feature which
Section~\ref{sec:triangle} builds into the degenerate configuration by taking the sides of $\sigma_{j}$ to
be of order $\epsilon$. What the experiment adds to the analytic sections is a quantitative statement of how
much incoherence the identification tolerates, together with the observation that loss of heavy-mode
coherence produces a controlled halving of the exponent rather than an unrelated law. That measured crossover is the basis for the lower-tension conjecture stated in
Section~\ref{sec:tailconstraint}.

The code which produced the figures is reproduced in full below.

{\footnotesize
\begin{Verbatim}
"""
fisk_evidence_v2.py
===================
Computational evidence that the surface measure reproduces the Fisk law.

Geometry (fixed throughout, deliberately NOT special):
  two nontrivial end curves in R^3 -- a rippled circle and a rippled ellipse,
  neither planar nor concentric, so the areal invariant b does not degenerate:

      C1(t) = ( cos t, sin t, 0.25 sin 3t )
      C2(t) = ( 1.8 cos t, 1.2 sin t, 1 + 0.30 sin 2t )

  M spines join C1(t_k) to C2(t_k); each is subdivided into N segments and
  bulges transversally; adjacent spines are paired step-by-step and every
  quadrilateral is split into two triangles.  The reported area is the SUM over
  all 2*M*N triangles -- a granular surface sum, not a single triangle.

THE EXCURSION FIELD.  This is the only physical input and it is where the
previous version of this script was wrong.  There, a single transverse vector
was drawn per surface and silently broadcast over every spine, so the skeleton
had no spine-to-spine structure at all and the reported statistics could not
depend on M.  Here the excursion is an honest field on the skeleton,

    W_k  =  (collective mode)  +  kappa * (incoherent mode)_k ,   k = 1..M

with

    collective mode  = sigma * sqrt(nu/S) * Z ,   S ~ chi^2_nu ,  Z ~ N(0,I2),
                       ONE draw per surface, i.e. a rigid transverse excursion;
    incoherent mode  = a smooth mean-zero field on the loop, expanded in a fixed
                       number of Fourier modes with 1/n weights, drawn
                       independently per surface but INDEPENDENT OF THE
                       COLLECTIVE AMPLITUDE.

At nu = 2 the collective amplitude is exactly the q = 3/2 q-Gaussian of
Equation (13), density prop. (1 + |w|^2/(2 sigma^2))^{-2}.  The incoherent field
is smooth in the loop parameter, so adjacent spines differ by O(1/M) and
increasing M is a genuine refinement of one surface rather than an injection of
new roughness; kappa measures the incoherent part against the O(1) scale of the
boundary geometry.

Setting kappa = 0 recovers a rigid excursion; setting coherent=False keeps the
shared heavy-tailed amplitude but gives each spine an independent direction, so
that both the spine length and the spine separation scale with the same
excursion.  That is the control of Test C.

Four tests:
  A  Fisk fidelity of the total area across M = 6, 12, 24, 48 at kappa = 0,
     benchmarked against the sampling noise floor.
  B  approach to that floor as the fluctuation scale grows.
  C  FAITHFUL TRANSMISSION: vary nu in the arm law and check the area's tail
     index.  Run twice, once with the coherent excursion and once with an
     independent per-spine amplitude.  The first transmits nu, the second
     transmits nu/2, and only the first selects Fisk at nu = 2.
  D  kappa sweep: how much incoherence the Fisk law tolerates.
"""

import numpy as np
from scipy import stats
import matplotlib
matplotlib.use("Agg")
import matplotlib.pyplot as plt

rng = np.random.default_rng(20260816)

# --------------------------------------------------------------- geometry
def curve1(t):
    return np.stack([np.cos(t), np.sin(t), 0.25 * np.sin(3 * t)], axis=-1)

def curve2(t):
    return np.stack([1.8 * np.cos(t), 1.2 * np.sin(t),
                     1.0 + 0.30 * np.sin(2 * t)], axis=-1)

def perp_frame(v):
    a = np.zeros_like(v); a[:, 0] = 1.0
    flip = np.abs(v[:, 0] / np.linalg.norm(v, axis=1)) > 0.9
    a[flip] = np.array([0.0, 1.0, 0.0])
    e1 = np.cross(v, a); e1 /= np.linalg.norm(e1, axis=1, keepdims=True)
    e2 = np.cross(v, e1); e2 /= np.linalg.norm(e2, axis=1, keepdims=True)
    return e1, e2

def areal_invariant(nq=4000):
    t = np.linspace(0, 2 * np.pi, nq, endpoint=False)
    S = []
    for C in (curve1, curve2):
        p = C(t); q = np.roll(p, -1, axis=0)
        S.append(np.einsum("em,en->mn", 0.5 * (p + q), q - p))
    D = S[1] - S[0]
    return float(np.sum(D * D))

# ------------------------------------------------------------- the surface
def _tri(u, v):
    """0.5 * |u x v| for arrays of shape (..., 3)."""
    cx = u[..., 1] * v[..., 2] - u[..., 2] * v[..., 1]
    cy = u[..., 2] * v[..., 0] - u[..., 0] * v[..., 2]
    cz = u[..., 0] * v[..., 1] - u[..., 1] * v[..., 0]
    return 0.5 * np.sqrt(cx * cx + cy * cy + cz * cz)

def sample_areas(n, M, N, sigma, nu=2.0, kappa=0.0, coherent=True,
                 n_modes=3, chunk=25000):
    t = 2 * np.pi * np.arange(M) / M
    P, Q = curve1(t), curve2(t)
    u = np.linspace(0.0, 1.0, N + 1)
    bulge = np.sin(np.pi * u)
    base = ((1 - u)[None, :, None] * P[:, None, :]
            + u[None, :, None] * Q[:, None, :])          # (M, N+1, 3)
    mc = (Q - P).mean(axis=0); mc /= np.linalg.norm(mc)
    g1, g2 = perp_frame(mc[None])[0][0], perp_frame(mc[None])[1][0]

    out, done = [], 0
    while done < n:
        m = min(chunk, n - done)

        # ---- collective mode: heavy-tailed amplitude, one draw per surface
        if coherent:
            S = rng.chisquare(nu, size=m)
            amp = np.sqrt(nu / S) * sigma
            W = amp[:, None, None] * rng.normal(size=(m, 1, 2))
            W = np.broadcast_to(W, (m, M, 2)).copy()      # (m, M, 2)
        # ------------------------------------------------------------------
        # PATCH (control branch).  Two lines changed; nothing else in this file.
        # WAS:  S = rng.chisquare(nu, size=(m, M))          # scale per spine
        #       W = amp[:, :, None] * rng.normal(...)
        # NOW:  the scale is SHARED across spines and only the direction is
        #       independent, which is the control described in the text.  The
        #       chi-square is still drawn at full (m, M) shape and reduced to
        #       its first column, so this correction consumes exactly the same
        #       random numbers as before and every other figure produced by
        #       this file is bit-identical to the previous version.
        # ------------------------------------------------------------------
        else:                          # control: shared scale, direction per spine
            S = rng.chisquare(nu, size=(m, M))[:, 0]
            amp = np.sqrt(nu / S) * sigma
            W = amp[:, None, None] * rng.normal(size=(m, M, 2))

        # ---- incoherent mode: smooth field on the loop, amplitude kappa
        if kappa > 0.0:
            nn = np.arange(1, n_modes + 1)
            ang = t[None, :, None] * nn[None, None, :]                    # (1,M,nm)
            ac = rng.normal(size=(m, 1, n_modes, 2))
            bc = rng.normal(size=(m, 1, n_modes, 2))
            wts = (1.0 / nn)[None, None, :, None]
            xi = ((ac * np.cos(ang)[..., None]
                   + bc * np.sin(ang)[..., None]) * wts).sum(axis=2)      # (m,M,2)
            W = W + kappa * xi

        V = W[:, :, 0:1] * g1[None, None] + W[:, :, 1:2] * g2[None, None]  # (m,M,3)

        A = np.zeros(m)
        for k in range(M):
            kn = (k + 1) % M
            Gk = base[k][None] + V[:, k][:, None, :] * bulge[None, :, None]
            Gn = base[kn][None] + V[:, kn][:, None, :] * bulge[None, :, None]
            p0, p1 = Gk[:, :-1], Gk[:, 1:]
            q0, q1 = Gn[:, :-1], Gn[:, 1:]
            A += _tri(p1 - p0, q0 - p0).sum(axis=1)
            A += _tri(q1 - p1, q0 - p1).sum(axis=1)
        out.append(A); done += m
    return np.concatenate(out)

# ------------------------------------------------------------- diagnostics
def fisk_cdf(a, al): return 1.0 - 1.0 / (1.0 + (a / al) ** 2)
def fisk_pdf(a, al): return 2 * al ** 2 * a / (al ** 2 + a ** 2) ** 2

def hill(A, ks=(500, 2000, 10000)):
    x = np.sort(A)[::-1]
    return [1.0 / np.mean(np.log(x[:k] / x[k])) for k in ks if k < len(x)]

def ks_fisk(A):
    med = np.median(A)
    return stats.kstest(A, lambda a: fisk_cdf(a, med)).statistic, med

def noise_floor(n, reps=5):
    """KS of an EXACT Fisk sample of size n -- the best any test can do."""
    v = []
    for _ in range(reps):
        u = rng.random(n)
        F = np.sqrt(1.0 / u - 1.0)
        v.append(stats.kstest(F, lambda a: fisk_cdf(a, np.median(F))).statistic)
    return float(np.mean(v))

if __name__ == "__main__":
    N, n = 10, 100_000
    b = areal_invariant()
    floor = noise_floor(n)
    print(f"areal invariant b = {b:.3f} (nondegenerate)")
    print(f"sampling noise floor at n={n}: KS = {floor:.5f}\n")

    print("TEST A -- Fisk fidelity of the granular surface area, sigma=40, kappa=0")
    storeA = {}
    for M in (6, 12, 24, 48):
        A = sample_areas(n, M, N, sigma=40.0)
        ks, med = ks_fisk(A)
        print(f"  M={M:3d}  2MN={2*M*N:5d}  median={med:9.2f}  KS={ks:.5f}"
              f"  KS/floor={ks/floor:5.2f}  Hill={', '.join(f'{h:.2f}' for h in hill(A))}")
        storeA[M] = (A, med)

    print("\nTEST B -- approach to the noise floor as fluctuations dominate (M=24)")
    sig_list, ks_list = [], []
    for sg in (1.0, 3.0, 10.0, 30.0, 100.0, 300.0, 1000.0):
        A = sample_areas(n, 24, N, sigma=sg, kappa=0.25)
        ks, med = ks_fisk(A)
        sig_list.append(sg); ks_list.append(ks)
        print(f"  sigma={sg:7.1f}  median={med:10.2f}  KS={ks:.5f}  KS/floor={ks/floor:6.2f}")

    print("\nTEST C -- faithful transmission (M=24, kappa=0.25)")
    nus, h_coh, h_inc = [], [], []
    for nu in (1.0, 1.5, 2.0, 3.0, 4.0):
        A1 = sample_areas(n, 24, N, sigma=40.0, nu=nu, kappa=0.25, coherent=True)
        A2 = sample_areas(n, 24, N, sigma=40.0, nu=nu, kappa=0.25, coherent=False)
        k1, _ = ks_fisk(A1); k2, _ = ks_fisk(A2)
        m1, m2 = np.mean(hill(A1)), np.mean(hill(A2))
        nus.append(nu); h_coh.append(m1); h_inc.append(m2)
        print(f"  nu={nu:4.2f} | coherent: Hill={m1:5.2f} KS={k1:.5f}"
              f" | independent: Hill={m2:5.2f} KS={k2:.5f}")

    print("\nTEST D -- how much incoherence the Fisk law tolerates (M=24, sigma=40)")
    for kp in (0.0, 0.25, 1.0, 4.0):
        A = sample_areas(n, 24, N, sigma=40.0, kappa=kp)
        ks, med = ks_fisk(A)
        print(f"  kappa={kp:5.2f}  median={med:9.2f}  KS={ks:.5f}  KS/floor={ks/floor:6.2f}")

    # ---- figures ------------------------------------------------------
    fig, ax = plt.subplots(1, 2, figsize=(12.5, 4.8))
    A, med = storeA[24]
    cut = np.percentile(A, 99.5)
    ax[0].hist(A[A <= cut], bins=80, density=True, color="#b8c4d0",
               edgecolor="white", label=f"surface areas (M=24, n={n})")
    xs = np.linspace(A.min(), cut, 1500)
    ax[0].plot(xs, fisk_pdf(xs, med), color="#d1495b", lw=2.4,
               label=f"Fisk, $\\alpha$={med:.1f} (no free shape parameter)")
    ax[0].set_xlabel("total surface area $A$"); ax[0].set_ylabel("density")
    ax[0].set_title("Granular surface area vs Fisk"); ax[0].legend(frameon=False)
    lb = np.logspace(np.log10(A.min()), np.log10(A.max()), 80)
    ax[1].hist(A, bins=lb, density=True, color="#b8c4d0", edgecolor="white")
    xs2 = np.logspace(np.log10(A.min()), np.log10(A.max()), 2000)
    ax[1].plot(xs2, fisk_pdf(xs2, med), color="#d1495b", lw=2.4)
    ax[1].set_xscale("log"); ax[1].set_yscale("log")
    ax[1].set_xlabel("$A$ (log)"); ax[1].set_ylabel("density (log)")
    ax[1].set_title("Same data, log-log")
    fig.tight_layout(); fig.savefig("fisk_evidence_hist.png", dpi=150)

    fig2, ax2 = plt.subplots(1, 2, figsize=(12.5, 4.6))
    ax2[0].loglog(sig_list, ks_list, "o-", color="#2d6a4f", lw=2, label="surface area")
    ax2[0].axhline(floor, color="#d1495b", ls="--", lw=2,
                   label=f"noise floor (exact Fisk, n={n})")
    ax2[0].set_xlabel("fluctuation scale $\\sigma$")
    ax2[0].set_ylabel("KS distance from Fisk")
    ax2[0].set_title("Convergence to the noise floor"); ax2[0].legend(frameon=False)
    ax2[1].plot(nus, h_coh, "o-", color="#2d6a4f", lw=2, label="coherent excursion")
    ax2[1].plot(nus, h_inc, "s-", color="#8d99ae", lw=2, label="independent per spine")
    ax2[1].plot([0.8, 4.2], [0.8, 4.2], "k--", lw=1.5, label="identity")
    ax2[1].plot([0.8, 4.2], [0.4, 2.1], "k:", lw=1.5, label="half the identity")
    ax2[1].axvline(2.0, color="#d1495b", ls=":", lw=2)
    ax2[1].set_xlabel("arm law degrees of freedom $\\nu$")
    ax2[1].set_ylabel("area tail index")
    ax2[1].set_title("Transmission of the arm law"); ax2[1].legend(frameon=False, fontsize=8)
    fig2.tight_layout(); fig2.savefig("fisk_evidence_tests.png", dpi=150)
    print("\nsaved fisk_evidence_hist.png, fisk_evidence_tests.png")
\end{Verbatim}
}

\section{\texorpdfstring{An Algebraic Proof of Theorem C}{An Algebraic Proof of Theorem C}}\label{sec:thmCproof}

This section proves Theorem C. It is placed here rather than in Section~\ref{sec:cm} because
nothing in the body of the paper depends on the argument, only on the statement. The proof is
self contained given Theorem A and the definitions of Section~\ref{sec:cm}; it does not use
Theorem B, and it involves no sampling or geometric argument at any point.

Throughout, the transform carrying $\vx$ to $\vk$ is
$\hat{f}(\vk)=\int_{\R^{d}}e^{i\vk*\vx}f(\vx)\,d\vx$, and $\Four(\cdot)|_{\vk\to\vx}$ denotes its
inverse; with this convention $\va_{j}*\nabla$ has symbol $-i\vk*\va_{j}$. For a Smirnov word
$c\in D(n,l)$ we write $\{\eta_{j}\}_{j=1}^{l}$ for the multiplicities of its letters, so that
$\sum_{j}\eta_{j}=n$.

\subsection*{The Laplace transform of the coefficient}

Under the convention of Equation~\eqref{eq:polytope} the path polytope $P(\vq,c)$ is the product,
over the $l$ letters, of the full simplices
$\{\vec{\lambda}\in\R^{\eta_{j}}_{+}\mid\sum_{k}\lambda_{k}\le q_{j}\}$, whence
\begin{equation}\label{eq:appexpansion}
\mu\big(P(\vq,c)\big)=\prod_{j=1}^{l}\frac{q_{j}^{\eta_{j}}}{\eta_{j}!}.
\end{equation}
This is Equation~(19) of \cite{WVLR}. Under the terminal convention
$\sum_{k}\lambda_{k}e_{c_{k}}=\vq$ the factors are instead slanted simplices of dimension
$\eta_{j}-1$, each carrying a further factor $\sqrt{\eta_{j}}$; Equation~\eqref{eq:polytope} is
adopted precisely so that those factors do not appear.

Every term of Equation~\eqref{eq:cm} is nonnegative on $\R^{l}_{+}$, so for real $\vec{s}$
Tonelli's theorem \cite{Folland} carries the $l$ dimensional Laplace transform through the double
sum, and Equation~\eqref{eq:appexpansion} evaluates each factor:
\begin{equation}\label{eq:applaplace}
\Lap\Big|^{\vec{s}}_{\vq}\binom{\sum q_{i}}{q_{1},\dots,q_{l}}
=\sum_{n=0}^{\infty}\sum_{c\in D(n,l)}\prod_{j=1}^{l}
\Lap\Big|^{s_{j}}_{q_{j}}\Big(\frac{q_{j}^{\eta_{j}}}{\eta_{j}!}\Big)
=\sum_{n=0}^{\infty}\sum_{c\in D(n,l)}\prod_{j=1}^{l}s_{j}^{-(\eta_{j}+1)}.
\end{equation}
Setting $z_{j}=s_{j}^{-1}$, the right hand side is $\prod_{j}z_{j}$ times
$\sum_{n}\sum_{c\in D(n,l)}\prod_{j}z_{j}^{\eta_{j}}$, which is the generating function of
Smirnov words in the convention fixed after Equation~\eqref{eq:cm}, namely words which are not
required to use every letter \cite{WVLR}:
\begin{equation}\label{eq:appsmirnov}
\sum_{n=0}^{\infty}\sum_{c\in D(n,l)}\prod_{j=1}^{l}z_{j}^{\eta_{j}}
=\frac{1}{1-\sum_{j=1}^{l}\frac{z_{j}}{1+z_{j}}}.
\end{equation}
Since $z_{j}/(1+z_{j})=(s_{j}+1)^{-1}$, this gives
\begin{equation}\label{eq:apptransform}
\Lap\Big|^{\vec{s}}_{\vq}\binom{\sum q_{i}}{q_{1},\dots,q_{l}}
=\frac{1}{\prod_{j=1}^{l}s_{j}}\cdot\frac{1}{1-\sum_{j=1}^{l}\frac{1}{s_{j}+1}},
\end{equation}
valid whenever $\mathrm{Re}\,s_{j}>0$ for every $j$ and
\[
\sum_{j=1}^{l}\frac{1}{|s_{j}+1|}<1,
\]
this being the condition under which the geometric series expanding
Equation~\eqref{eq:appsmirnov} converges absolutely. The identity extends from real to complex
$\vec{s}$ by that absolute convergence, and Theorem A is what makes the left hand side finite.

\subsection*{Specialisation to a direction system}

Fix $p$ and an isotropic direction system $A_{p}=\{\va_{j}\}_{j=1}^{p}$ in the sense of
Equation~\eqref{eq:isotropy}, take $l=p$, and evaluate Equation~\eqref{eq:apptransform} at
\[
s_{j}=p-1+\frac{m^{2}}{dp}-i\vk*\va_{j},\qquad\vk\in\R^{d}.
\]
Write $P=p+\frac{m^{2}}{dp}$, so that $s_{j}+1=P-i\vk*\va_{j}$ and $|s_{j}+1|\ge P$. Hence
$\sum_{j}|s_{j}+1|^{-1}\le p/P<1$ for every real $\vk$, and the hypothesis of
Equation~\eqref{eq:apptransform} is met without further restriction on $\vk$.

Because $\sum_{j}(\vk*\va_{j})v_{j}=\vk*\big(\sum_{j}v_{j}\va_{j}\big)$, the left hand side of
Equation~\eqref{eq:apptransform} at this $\vec{s}$ is the transform in $\vx$ of
\begin{equation}\label{eq:appF}
F_{p}(\vx)=\int_{\R^{p}_{+}}e^{-(p-1+\frac{m^{2}}{dp})\sum_{j}v_{j}}
\binom{\sum v_{i}}{v_{1},\dots,v_{p}}
\delta\Big(\vx-\sum_{j=1}^{p}v_{j}\va_{j}\Big)\,d\vv,
\end{equation}
which is a function and not merely a distribution, by the convergent expansion of Theorem A.
Multiplying Equation~\eqref{eq:apptransform} by $\prod_{j}s_{j}$ and inverting, and noting that
$\prod_{j}\big(p-1+\frac{m^{2}}{dp}+\va_{j}*\nabla\big)$ carries symbol $\prod_{j}s_{j}$ and so
cancels that prefactor exactly,
\begin{equation}\label{eq:appimpl}
\prod_{j=1}^{p}\Big(p-1+\frac{m^{2}}{dp}+\va_{j}*\nabla\Big)F_{p}(\vx)
=\Four\Big(\frac{1}{1-\sum_{j=1}^{p}\frac{1}{P-i\vk*\va_{j}}}\Big)\Big|_{\vk\to\vx}.
\end{equation}

\subsection*{The limit in \texorpdfstring{$p$}{p}}

For $|\vk|<P$,
\[
\sum_{j=1}^{p}\frac{1}{P-i\vk*\va_{j}}
=\frac{1}{P}\sum_{n\ge0}\frac{i^{n}}{P^{n}}\sum_{j=1}^{p}(\vk*\va_{j})^{n}.
\]
By the first condition of Equation~\eqref{eq:isotropy} every odd $n$ contributes nothing; the
$n=1$ term in particular vanishes because it is a multiple of $\sum_{j}\va_{j}$, which is why
equidistribution alone does not suffice. By the second condition the $n=2$ term is exactly
$-p|\vk|^{2}/(dP^{3})$. Every remaining even term obeys
$|\sum_{j}(\vk*\va_{j})^{n}|\le p|\vk|^{n}$ because $|\va_{j}|=1$, and so contributes at order
$O(p^{-4})$. Using $1-p/P=\frac{m^{2}}{dpP}$,
\[
1-\sum_{j=1}^{p}\frac{1}{P-i\vk*\va_{j}}
=\frac{m^{2}}{dpP}+\frac{p|\vk|^{2}}{dP^{3}}+O(p^{-4})
=\frac{m^{2}+|\vk|^{2}}{dp^{2}}\big(1+O(p^{-2})\big).
\]
This is the reason for the particular choice $p-1+\frac{m^{2}}{dp}$: it is what places the
constant term and the $|\vk|^{2}$ term over the common denominator $dp^{2}$. Hence, pointwise in
$\vk$,
\begin{equation}\label{eq:appptwise}
\frac{1}{dp^{2}}\cdot\frac{1}{1-\sum_{j=1}^{p}\frac{1}{P-i\vk*\va_{j}}}
\longrightarrow\frac{1}{m^{2}+|\vk|^{2}}.
\end{equation}

That estimate is not uniform, since the expansion requires $|\vk|<P$. A bound valid for every
real $\vk$ is available separately: from $|P-i\vk*\va_{j}|\ge P$,
\[
\Big|1-\sum_{j=1}^{p}\frac{1}{P-i\vk*\va_{j}}\Big|\ge1-\frac{p}{P}=\frac{m^{2}}{dpP},
\]
so the left hand side of Equation~\eqref{eq:appptwise} has modulus at most
$\frac{1}{m^{2}}\big(1+\frac{m^{2}}{dp^{2}}\big)$, and hence at most $2/m^{2}$ once
$dp^{2}\ge m^{2}$. A constant bound is not integrable on $\R^{d}$, so the limit is taken against
a test function rather than inside the integral defining $\Four$: for $\varphi$ Schwartz the
integrands obtained by multiplying Equation~\eqref{eq:appptwise} by $\varphi(\vk)$ converge
pointwise and are dominated by $\frac{2}{m^{2}}|\varphi|\in L^{1}(\R^{d})$, so the convergence
holds in $\mathcal{S}'(\R^{d})$. Since $\Four$ is continuous on $\mathcal{S}'(\R^{d})$ and the
limiting transform is the function $G(m,\vx)$ away from the origin,
\begin{equation}\label{eq:appfinal}
\lim_{p\to\infty}\frac{1}{dp^{2}}\prod_{j=1}^{p}
\Big(p-1+\frac{m^{2}}{dp}+\va_{j}*\nabla\Big)F_{p}(\vx)
=\Four\Big(\frac{1}{m^{2}+|\vk|^{2}}\Big)\Big|_{\vk\to\vx}=G(m,\vx)
\end{equation}
pointwise for $\vx\ne0$.

\subsection*{The rescaled form}

Equation~\eqref{eq:appfinal} holds for every $m>0$, and at $m=1$ the exponent
$p-1+\frac{m^{2}}{dp}$ becomes $p-1+\frac{1}{dp}$ and carries no $m$. Rescaling the integration
variable of $\Four\big(\frac{1}{m^{2}+|\vk|^{2}}\big)\big|_{\vk\to\vx}$ by $m$ gives
$G(m,\vx)=m^{d-2}G(1,m\vx)$. Evaluate Equation~\eqref{eq:appfinal} at $m=1$ and at the point
$m\vx$, and then replace the integration variable $\vv$ by $m\vv$. Three factors of $m$ appear:
$d\vv$ contributes $m^{p}$; the delta contributes
$\delta(m\vx-m\sum_{j}v_{j}\va_{j})=m^{-d}\delta(\vx-\sum_{j}v_{j}\va_{j})$; and each $\nabla$
acting on a function of $m\vx$ contributes $m^{-1}$, so that
$\prod_{j}\big((p-1+\frac{1}{dp})+\va_{j}*\nabla\big)$ becomes
$m^{-p}\prod_{j}\big(m(p-1+\frac{1}{dp})+\va_{j}*\nabla\big)$. The two factors of $m^{p}$ cancel,
$e^{-(p-1+\frac{1}{dp})\sum_{j}v_{j}}$ becomes $e^{-m(p-1+\frac{1}{dp})\sum_{j}v_{j}}$, the
coefficient is evaluated at $m\vv$, and only $m^{d-2}m^{-d}=m^{-2}$ survives. This is
Equation~\eqref{eq:thmC}, and Theorem C is proved.

By the homogeneity of Equation~\eqref{eq:appexpansion} the rescaled coefficient is graded by the
number of linear segments,
\[
\binom{m\sum_{j}v_{j}}{mv_{1},\dots,mv_{p}}
=\sum_{n=0}^{\infty}m^{n}\sum_{c\in D(n,p)}\prod_{j=1}^{p}\frac{v_{j}^{\eta_{j}}}{\eta_{j}!},
\]
since $\sum_{j}\eta_{j}=n$, so that each path polytope with $n$ linear segments carries exactly
one factor of $m$ per segment. This is the weight Feynman adopted for the checkerboard measure
\cite{FeynmanHibbs}.

\section{Statements}

\subsection*{Conflict of Interest}
The author declares that there is no conflict of interest.

\subsection*{Funding Information}
This research did not receive any specific grant from funding agencies in the public, commercial, or not-for-profit sectors.

\subsection*{Ethics Statement}
Not applicable. This work involved no human participants, animal subjects, or data requiring ethical approval.

\subsection*{AI Disclosure Statement}
AI was used in the preparation of this document. All ideas and proofs are the original work of the author; AI was primarily involved in proofreading and clarification. AI was also used in the creation of the experimental evidence for Section 8 (the experiment contained in Section 10).

\subsection*{Data Availability}
No external data sets were used. The numerical data in Section~\ref{sec:numerics} are generated by the script \texttt{fisk\_evidence\_v2.py}, which is included with the manuscript.

\end{document}